\documentclass[journal,twoside,final]{IEEEtran}
\usepackage{amsthm,amssymb,graphicx,multirow,amsmath,color,amsfonts}
\usepackage[update,prepend]{epstopdf}
\usepackage[noadjust]{cite}
\usepackage[latin1]{inputenc}

\pdfoutput=1

\usepackage{pdfpages}
\usepackage{flushend}
\usepackage{tabulary}
\usepackage{multirow}
\usepackage{url}

\catcode`,\active

\catcode`\,12

\def\nb0{{\mathbf{0}}}
\def\nb1{{\mathbf{1}}}

\def\ncalA{{\mathcal{A}}}
\def\ncalB{{\mathcal{B}}}
\def\ncalC{{\mathcal{C}}}

\def\ncalE{{\mathcal{E}}}

\def\ncalK{{\mathcal{K}}}
\def\ncalL{{\mathcal{L}}}

\def\ncalO{{\mathcal{O}}}

\def\nbbE{{\mathbb{E}}}

\def\nbbP{{\mathbb{P}}}

\def\nbbR{{\mathbb{R}}}

\def\nrmd{{\rm d}}

\newtheorem{lemma}{Lemma}

\newtheorem{ndef}{Definition}

\newtheorem{theorem}{Theorem}
\newtheorem{prop}{Proposition}
\newtheorem{cor}{Corollary}

\newtheorem{remark}{Remark}

\def\E{\mathbb{E}}

\def\P{\mathbb{P}}
\def\pc{\mathtt{P_c}}
\def\rc{\mathtt{R_c}}   

\def\R{\mathbb{R}}

\def\T{\beta}							
\def\snr{\mathtt{SNR}}
\def\sir{\mathtt{SIR}}

\allowdisplaybreaks 

\begin{document}

\graphicspath{{./Figures/}}
\title{Downlink MIMO HetNets: Modeling, Ordering Results and Performance Analysis}
\author{Harpreet S. Dhillon,~\IEEEmembership{Student Member, IEEE}, Marios Kountouris,~\IEEEmembership{Member, IEEE}, \\ and Jeffrey G. Andrews,~\IEEEmembership{Fellow, IEEE}
\thanks{This work was supported by NSF grant CIF-1016649. A part of this paper was presented at Asilomar 2012~\cite{DhiKouC2012}.}
\thanks{H. S. Dhillon and J. G. Andrews are with the Wireless Networking and Communications Group (WNCG), The University of Texas at Austin, TX, USA (email: dhillon@utexas.edu and jandrews@ece.utexas.edu).} 
\thanks{M. Kountouris is with the Department of Telecommunications, SUPELEC (Ecole Sup{\'e}rieure d'Electricit{\'e}), France (email: marios.kountouris@supelec.fr). \hfill Manuscript updated: \today.
} 
}

\maketitle

\begin{abstract}

We develop a general downlink model for multi-antenna heterogeneous cellular networks (HetNets), where base stations (BSs) across tiers may differ in terms of transmit power, target signal-to-interference-ratio ($\sir$), deployment density, number of transmit antennas and the type of multi-antenna transmission. In particular, we consider and compare space division multiple access (SDMA), single user beamforming  (SU-BF), and baseline single-input single-output (SISO) transmission. For this general model, the main contributions are: (i) ordering results for both coverage probability and per user rate in closed form for any BS distribution for the three considered techniques, using novel tools from stochastic orders, (ii) upper bounds on the coverage probability assuming a Poisson BS distribution, and (iii) a comparison of the area spectral efficiency (ASE).  The analysis concretely demonstrates, for example, that for a given total number of transmit antennas in the network, it is preferable to spread them across many single-antenna BSs vs. fewer multi-antenna BSs.  Another observation is that SU-BF provides higher coverage and per user data rate than SDMA, but SDMA is in some cases better in terms of ASE.  

\end{abstract}


\section{Introduction}
Cellular networks are undergoing a significant makeover to handle increasing data demands prompted by the rising popularity of data hungry applications, such as video streaming and real time video calls~\cite{CisM2012}. Handling these changing usage trends requires much higher data rates than contemporary cellular networks were designed for. Two strategies that stand out to provide such spectral efficiencies are: i) deploy low power nodes to reduce frequency reuse distance, and ii) equip BSs with multiple antennas to enable the use of multiple antenna techniques, such as beamforming and SDMA. Multiple antenna techniques are already relatively mature, being part of multiple wireless standards such as 
IEEE 802.11e WiMAX and 3GPP LTE-A~\cite{LiLiJ2010}, apart from plethora of theoretical research activities in academia~\cite{BolGesB2006}. Similarly, the concept of deploying low power nodes, also termed as small cells or HetNets, has been researched both in industry and academia for a fairly long time, see for example~\cite{MadBorJ2010,AndClaJ2012} and the references therein. The standardization activities for HetNets have also started in 3GPP release 10~\cite{3GPP2012}. These activities clearly indicate that multi-antenna techniques and HetNets will coexist and complement each other in the future wireless networks and should not be studied in isolation, as has been typically done in the literature. In this paper, we address this problem and develop a general tractable model and the corresponding analytical tools for multi-antenna HetNets using techniques from stochastic orders and stochastic geometry.

\subsection{Related Work}
As cellular networks have become denser, they have also become increasingly irregular. This is particularly true for small cells which are deployed opportunistically and in hotspots, and thus highly irregular. As a result, the popular deterministic grid model is increasingly anachronistic for ongoing and future deployments. Even for single-tier networks, the grid model is quite idealized and a perturbed grid model is sometimes used for macrocell locations~\cite{NieLemC2003,MitRosC2012}. 

Although more data is needed to make conclusive statements on which is a better model for HetNets, it seems to many that a random spatial model will often be a more appropriate model versus a deterministic one. In a random spatial model, the BS locations are modeled by a two-dimensional point process, the simplest being the Poisson Point Process (PPP)~\cite{AndBacJ2011,DhiGanJ2012}. This model has the advantages of being scalable to multiple classes of overlaid BSs and accurate to model location randomness, especially that of the small cells. Additionally, powerful tools from stochastic geometry can be used to derive performance results for general multi-tier networks in closed form, which was not even possible for single-tier networks using deterministic grid model~\cite{DhiGanJ2012}. While sufficient progress has been made in modeling single-antenna (SISO) HetNets~\cite{DhiGanJ2012,JoSanJ2012,MukJ2012,MadResC2011}, the efforts to understand multi-antenna HetNets have just begun, e.g., see~\cite{HeaKouJ2013}. 

The main challenge in modeling multi-antenna HetNets is the number of possible multi-antenna techniques to choose from in each tier along with their tractable characterization. As a result, most prior works on multi-antenna HetNets have focused only on two-tier networks. For this paper, the most relevant one is \cite{ChaKouAnd2009}, where SU-BF was shown to achieve better coverage than multiuser linear beamforming on the downlink of femtocell-aided cellular network assuming perfect channel state information (CSI). Random orthogonal beamforming with max-rate scheduling and coordinated beamforming for femtocell underlay networks was analyzed in \cite{Park2010}, \cite{Park2011}, respectively. The effect of channel uncertainty on linear beamforming in two-tier networks was investigated in \cite{AkoKouHea11}. 
In addition to the contributions in cellular networks, there has been extensive work on analyzing multi-antenna techniques in wireless ad hoc networks, which is also related to our work since several tools and techniques developed therein can be employed and extended to HetNets. Several single-user MIMO techniques, such as spatial diversity, open loop transmission and spatial multiplexing, have been studied, see for instance \cite{HunWebAnd08, LouKayCol11,VazHea12}. The performance of multiuser MIMO communication in a Poisson field of interferers, with perfect and quantized CSI at the transmitter was investigated in \cite{KouAnd09} and \cite{KouAnd12}, respectively.

\subsection{Contributions and Outcomes}
\subsubsection*{Downlink model for multi-antenna HetNets} In Section II, we develop a comprehensive downlink model consisting of $K$-tiers or classes of BSs, such as macrocells, femtocells, picocells and distributed antennas. The BSs across tiers differ in terms of transmit power, deployment density, target $\sir$, number of transmit antennas, number of users served in each resource block, and the type of multi-antenna transmission. 
We also consider the possibility of closed subscriber group or closed access in which a typical user is granted access to only a few BSs, while the rest purely act as interferers. 

\subsubsection*{Ordering results for coverage and rate} For general system models, such as the one considered in this paper, it is not always possible to express key performance metrics such as coverage probability and per user rate in closed form.  In the absence of simple analytical expressions, it is difficult to compare different transmission techniques in general HetNet settings. To facilitate this comparison, in Section III we derive {\em ordering results} for both the coverage probability and the rate per user in both open and closed access networks, using tools from stochastic orders. Interested readers can refer to \cite{TepRajJ2011,BlaYogJ2009,MadResJ2012,LeeTepJ2012} for application of stochastic orders to conventional wireless networks. While circumventing the need for deriving coverage and rate expressions, this analysis leads to several system design guidelines, e.g., it concretely demonstrates the superiority of serving a single user in each resource block, either by SISO or SU-BF, as opposed to serving multiple users by SDMA, both in terms of coverage and rate, under a per user power constraint. The BS locations for this analysis may be drawn from any general stationary point process, not necessarily independent across tiers, which is a significant generalization of earlier HetNet models based on the PPP assumption, such as the one considered in~\cite{DhiGanJ2012}.

\subsubsection*{Area spectral efficiency comparison} While comparison of different configurations of multi-antenna HetNets in terms of coverage probability and average rate per user is conclusive from the ordering results, it does not directly capture the fact that some transmission techniques serve more users than the others and hence provide higher sum data rate. In order to capture this effect, in Section IV we additionally consider ASE, which gives the number of bits transmitted per unit area per unit time per unit bandwidth. To facilitate the comparison of transmission techniques in terms of ASE, we first derive an upper bound on the coverage probability of a typical user in both open and closed access networks assuming that the BS locations for each tier are drawn from independent PPPs and show that it can be reduced to a closed form expression for the ``full'' SDMA case (where the number of users served is equal to the number of antennas). A part of this derivation is reported in the shorter version of this paper~\cite{DhiKouC2012}. The tightness of the bound is studied and it is shown that the closed form bound derived for full SDMA is tight down to very low target $\sir$s. Using this expression and the previously known coverage probability results for SISO HetNets~\cite{DhiGanJ2012}, we derive ASE results for various transmission techniques in closed form. Main consequences of this analysis are: i) for the same density of BSs, SISO HetNets have lower ASE than SDMA since they serve fewer users. Interestingly, despite serving fewer users, SU-BF outperforms SDMA in moderate and high target $\sir$ regime, and ii) when the BS densities are adjusted such that all the transmission techniques serve the same density of users, the ASEs of SU-BF, SISO and SDMA follow the same ordering as that of coverage probability and average rate per user.

\section{System Model \label{sec:sysmod}}

\begin{table}
\centering
\caption{Notation Summary}
\label{table:notationtable}
\begin{tabulary}{\columnwidth}{ |c | C | }
\hline
    \textbf{Notation} & \textbf{Description} \\ \hline
    $\Phi_k$					& A point process modeling the locations of $k^{th}$ tier BSs \\ \hline
    $\Phi_u$					& An independent PPP modeling user locations \\ \hline
    $P_k; \lambda_k$			& 	Downlink transmit power to each user; deployment density of the $k^{th}$ tier BSs \\ \hline
    $M_k, \Psi_k$				& Number of transmit antennas; number of users served in each resource block by a $k^{th}$ tier BS \\ \hline
    $h_{kx}$			& Channel power of the direct link from a $k^{th}$ tier BS located at $x$ to a typical user, $h_{kx} \sim \Gamma(\Delta_k, 1)$ with $\Delta_k = M_k - \Psi_k + 1$\\ \hline 
    $g_{jy}$			& Channel power of the interfering link from a $j^{th}$ tier BS located at $y$ to a typical user, $g_{jy} \sim \Gamma(\Psi_y, 1)$ \\ \hline 
    $\ncalK; \{x_k\}$ & $\{1, 2, \ldots, K\}; \{x_{k_1}, x_{k_1+1}, \ldots x_{k_2}\}$, where the values of $k_1$ and $k_2$ will be clear from the context \\ \hline
    $\ncalB$ 				& $\ncalB \subset \ncalK$ denotes the set of open access tiers \\ \hline
    $\pc; \beta_k$			& Coverage probability (in terms of $\sir$); target $\sir$ for $k^{th}$ tier \\ \hline
    $\rc; \ncalO_k, T_k$	& Rate coverage; fraction of resources allocated to each user served by $k^{th}$ tier; $k^{th}$ tier target rate \\ \hline
    $\eta$					&  Area spectral efficiency \\ \hline
    $Z_{k,m}$ 			& $Z_{k,m} = X_1/X_2$, where $X_1 \sim \Gamma(k,1)$ and $X_2 \sim \Gamma(m,1)$ \\ \hline
\end{tabulary}

\end{table}

\subsection{System Setup and BS Location Model}
We consider $K$ different classes or tiers of BSs, indexed by the set $\ncalK = \{1, 2, \ldots, K\}$. The BSs across tiers differ in terms of their transmit power $P_k$ with which they transmit to each user, deployment density $\lambda_k$, target $\sir$ $\T_k$, number of antennas $M_k$ and number of users served by each BS in a given resource block $\Psi_k \leq M_k$, $\forall k \in \ncalK$. While in open access networks a mobile user can connect to any BS, in closed access networks the access is restricted to $\ncalB \subset \ncalK$ tiers. For the ordering part (Section III), the locations of BSs of each tier are drawn from a general stationary point process $\Phi_k$. The point processes $\Phi_k$ are not necessarily independent. Those familiar with the recent advances in modeling cellular networks with random spatial models will immediately recognize that this is a significant generalization over the known models based on the PPP assumption. This is enabled by the fact that the ordering results are based on the ordering of the fading components of the channel power distributions for various setups and do not depend upon the spatial point process governing the locations of the BSs. However, we do require further assumptions for the ASE comparison, which involves the derivation of explicit expressions for coverage probabilities. Therefore, for the ASE analysis we will consider the more familiar independent PPP model, where each tier of BSs is modeled by an independent homogeneous PPP of density $\lambda_k$.

Nevertheless, this is a fairly general model that captures the current deployment trends in $4$G networks, e.g., it is easy to imagine hundreds of femtocells coexisting in each macrocell, transmitting at orders of magnitude lower power than macrocells, having relatively small number of antennas due to smaller form factor, serving smaller number of users due to smaller coverage footprints and providing restricted access to their own users due to a smaller backhaul capacity or privacy concerns. A two-tier illustration of the proposed system model is shown in Fig.~\ref{fig:SystemMod}, where a high power macro tier with four transmit antennas per BS coexists with a low power pico tier with two antennas per BS. Owing to its bigger coverage footprint, each macro BS serves higher load than its pico counterpart~\cite{DhiGanJ2013}. The coverage region of each BS in this illustration corresponds to the region where it provides the maximum average received power, thereby leading to a weighted Voronoi tessellation~\cite{DhiGanJ2012}. 

We model the user locations by an independent PPP $\Phi_u$ and focus on the downlink analysis performed at a single-antenna user located at the origin. This analysis at the origin is facilitated by Slivnyak's theorem, which states that the properties observed at a typical point of $\Phi_u$ are the same as those observed by the point at origin in the point process $\Phi_u \cup \{0\}$~\cite{StoKenB1995}. For the interference, we consider full-buffer model where the interfering BSs are assumed to be always transmitting~\cite{AndBacJ2011,DhiGanJ2012}. More sophisticated load models~\cite{DhiGanJ2013} along with non-uniform user distributions~\cite{DhiGanJ2013b} can also be considered but are out of the scope of this paper.

\begin{figure}[t!]
\centering
\includegraphics[width=.9\columnwidth]{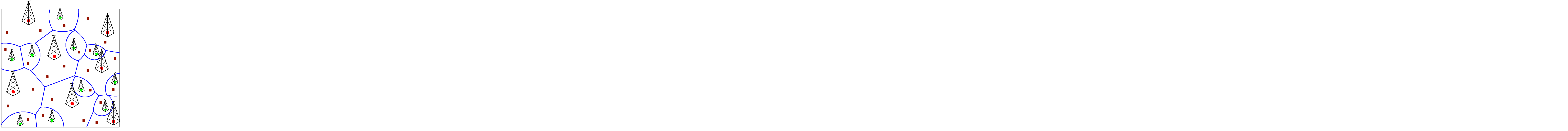}
\caption{An illustration of a possible two-tier multi antenna HetNet configuration, with four antenna macro BSs serving two users and two antenna pico BSs serving one user each. The circles, triangles and rectangles represent macro BSs, pico BSs, and mobile users, respectively.}
\label{fig:SystemMod}
\end{figure}

\subsection{Channel Model}
Before going into the technical details, it is important to understand that the channel power distribution of the link from a multi-antenna BS to a typical single-antenna user depends upon the transmission technique and whether it is a serving BS or an interferer. For example, if it is a serving BS, it may precode its signals for a typical user depending upon the transmission technique, which may lead to a different effective channel distribution from the case when it simply acts as an interferer. Therefore, to develop a general framework in which the BSs across tiers may differ in terms of the number of transmit antennas and the transmission technique, we assume that the channel power for the direct link from a $k^{th}$ tier BS located at $x_k \in \R^2$ to a typical user located at origin is denoted by $h_{kx_k}$ and for the interfering link from a $j^{th}$ tier BS located at $y \in \R^2$ is denoted by $g_{jy}$. In this paper, we assume perfect CSI and focus on zero-forcing precoding, under which for Rayleigh fading it can be argued that the channel power distributions of both the direct and the interfering links follow the Gamma distribution~\cite{HuaPapB2012}. 
As discussed in detail in Appendix~\ref{appendix:signal}, it can therefore be shown that for zero forcing $h_{kx} \sim \Gamma(\Delta_k, 1)$ and $g_{jy} \sim \Gamma(\Psi_j, 1)$, where $\Delta_k = M_k - \Psi_k + 1$. Note that for SISO transmission there is no precoding and hence the channel gains $h_{kx}$ and $g_{kx}$ from a BS to a typical user are the same. Under Rayleigh fading assumption, both follow $\exp(1)$ distribution, which is the same as $\Gamma(1,1)$ distribution. Other precoding techniques, such as minimum mean square error (MMSE), are left for future work.

Although for brevity we limit our discussion to Rayleigh fading channels, other fading distributions under which the channel power for both the desired and the interfering links follow Gamma distribution after precoding, e.g., Nakagami-$m$, can also be studied using the proposed techniques. The shape and the scale parameters for the Gamma distributions corresponding to the channel powers of the desired and interfering links can be derived using techniques well known in the literature, e.g., see~\cite{ZheKaiJ2006}. For concreteness, we will focus on the following three transmission techniques in this paper:
\begin{itemize}
\item {\em SDMA:} In this case, a $k^{th}$ tier BS with $M_k$ antennas serves $\Psi_k > 1$ users in each resource block. When $\Psi_k = M_k$, we term it as full SDMA, for which $\Delta_k=1$.
\item {\em SU-BF:} In this case, a $k^{th}$ tier BS serves $\Psi_k = 1$ users in each resource block.
\item {\em SISO:} Baseline single-antenna case~\cite{DhiGanJ2012}, where each BS serves one user in each resource block.
\end{itemize}

For each transmission technique, the received power at a typical single-antenna user located at origin from the BS located at $x_k\in \Phi_k$ is
\begin{align}
P_r = P_k h_{kx_k} \|x_k\|^{-\alpha},
\end{align}
where $\alpha$ is the path loss exponent, because we assume a per user power constraint in this formulation. Note that although shadowing is not explicitly considered, its effect on downlink performance can be easily incorporated by scaling the deployment density of each BS class with a constant factor that depends upon the shadowing distribution, as discussed in details in~\cite{DhiAndJ2013}. The received $\sir$ can now be expressed as
\begin{align}
\sir(x_k) = \frac{P_k h_{kx_k} \|x_k\|^{-\alpha}}{\sum_{k\in \ncalK} \sum_{y\in \Phi_j \setminus x_k} P_j g_{jy} \|y\|^{-\alpha}}.
\label{eq:SIR_def}
\end{align}
For notational simplicity, we assume that the thermal noise is negligible as compared to the self interference and is hence ignored. This is justified in the current wireless networks, which are typically interference limited~\cite{BouPanJ2009}. As will be evident from our analysis, thermal noise can be included in the proposed framework with little extra work. For cell association, we assume that the set of the candidate serving BSs is the collection of the BSs that provide the strongest instantaneous received power from each tier to which a typical mobile is allowed to connect. A typical user is said to be in coverage if the received $\sir$ from one of these candidate serving BSs is more than the respective target $\sir$, as discussed in detail in the next section. We conclude this section with a remark that although we consider perfect CSI, it is possible to relax this assumption to study the effect of imperfect CSI on the performance of multi-antenna HetNets, as discussed below.

\begin{remark}[Imperfect CSI]
Using tools from \cite{Jindal06, KouAnd12, ZhangKou:11}, it is possible to derive the received channel power and interference statistics for both SU-BF and SDMA under quantized channel directional information (CDI). In particular, in a system where each user reports CDI using $B$ feedback bits, the desired channel gain is exponentially distributed for both SU-BF and SDMA. However, in the latter case, the inter-user interference is not completely eliminated due to zero-forcing beamforming using imperfect CSI. Therefore, an additional interference term, independent of the multi-tier interference, appears in the denominator of~\eqref{eq:SIR_def}, which is distributed as $\Gamma(\Psi_k,\delta)$ with $\delta = 2^{-\frac{B}{M_k-1}}$ under quantization cell approximation~\cite{KouAnd12}.
\end{remark}

\section{Ordering Results for Coverage and Rate \label{sec:ordering}}
This is the first main technical section of this paper, where we compare the performance of various transmission techniques in terms of coverage probability and rate per user. We first study coverage probability in detail and then show that the analysis can be easily extended to study rate per user. We begin by formally defining the coverage probability. 

\begin{ndef}[Coverage probability] \label{def:pc}
A typical user in an open access network is said to be in coverage if its downlink $\sir$ from at least one of the BSs is higher than the corresponding target. This can be mathematically expressed as
\begin{align}
\pc = \P\left(\bigcup_{k\in \ncalK} \max_{x_k \in \Phi_k} \sir(x_k) > \T_k \right).
\end{align}
The coverage probability can be equivalently defined as the average area in coverage or the average fraction of users in coverage. For closed access networks, the definition remains the same, except that the union is now over the set of tiers $\ncalB \subset \ncalK$ to which a typical user is allowed to connect.
\end{ndef}

\begin{remark} [Open access vs. closed access coverage]
The coverage probability in open access networks is always higher than in closed access networks. It follows directly by definition of coverage probability along with the fact that for $\ncalB \subset \ncalK$
\begin{align}
&\nb1 \left( \bigcup_{k\in \ncalB} \max_{x_k \in \Phi_k} \sir(x_k) > \T_k \right) \nonumber \\
&\leq \nb1 \left( \bigcup_{k\in \ncalK} \max_{x_k \in \Phi_k} \sir(x_k) > \T_k\right),
\end{align}
where the indicator function $\nb1(\ncalE)$ is $1$ when event $\ncalE$ holds and $0$ otherwise.
\end{remark}


Owing to the complexity of the system model considered in this paper, it is not always possible to express coverage probability in simple closed form for any general system configuration, especially when the BS locations are drawn from a general point process. As evident from our analysis in the next section, this presents the first main challenge in comparing various transmission techniques. In this section, we take a slightly different view of this problem and focus on ``ordering'' the relative performance of different system configurations using tools from stochastic orders. Interested readers can refer to~\cite{ShaShaB2007} for details on stochastic orders. It is important to note that stochastic orders operate on random variables, as opposed to related majorization theory, which defines partial order on deterministic vectors~\cite{MarOlkB2009}.

This powerful approach allows insights into the relative performance of different transmission techniques, while circumventing the need to evaluate complicated expressions for the performance metrics such as coverage and rate. We begin by defining first order stochastic dominance as follows.

\begin{ndef}[First order stochastic dominance] For any two random variables (rvs) $Z_1$ and $Z_2$, $Z_1$ (first order) stochastically dominates $Z_2$ if and only if\begin{align}
\P[Z_1 > z] \geq \P[Z_2 > z],\ \forall z.
\end{align}
Equivalently, $Z_1$ is greater than $Z_2$ in the usual stochastic order and is denoted by $Z_1 \geq_{\rm st} Z_2$.
\end{ndef}

Therefore, $Z_1 \geq_{\rm st} Z_2$ if and only if the complementary cumulative distribution function (CCDF) of $Z_1$ dominates that of $Z_2$ over the whole range. It is intuitively clear and will be made precise later in this section that the proper understanding of the ordering of received $\sir$ for different system configurations plays a central role in studying their coverage and rate ordering. The main technical idea behind the proposed ordering approach is to condition on the distribution of the BS locations and then order the received $\sir$s for different transmission techniques by ordering the fading components of the channel powers of both the desired and the interfering links. This idea of ordering the channel power distributions has been previously used in the literature to compare the performance of wireless links in terms of signal-to-noise-ratio ($\snr$) and related metrics such as ergodic capacity and error rates for different modulation schemes, e.g., see~\cite{TepRajJ2011} and references therein. However, to the  best of our understanding, this approach has never been used for $\sir$ ordering in the context of HetNets. Now note that the received $\sir$ can be alternatively expressed as
\begin{align}
\sir(x_k) = \frac{P_k \|x_k\|^{-\alpha}}{\sum_{k\in \ncalK} \sum_{y\in \Phi_j \setminus x_k} P_j g_{jy}/h_{kx_k} \|y\|^{-\alpha}},
\label{eq:SIR_alternate}
\end{align}
where the randomness due to propagation channel is lumped into ratios of Gamma random variables $h_{kx_k}$ and $g_{jy}$. Therefore, it turns out that it is important first to understand the ordering of the ratios of Gamma random variables, which is studied next.

\subsection{Ordering of the Ratios of Gamma Random Variables}
For concreteness, define the ratio of two random variables $X_1\sim \Gamma(k, 1)$ and $X_2 \sim \Gamma(m, 1)$ by $Z_{k,m} = X_1/X_2$. It is easy to derive the cumulative distribution function (CDF) of $Z_{k,m}$ using basic algebraic manipulations and is given by
\begin{align}
F_{Z_{k,m}} (z) = 1 - \frac{1}{\Gamma(m)}\sum_{i=0}^{k-1} \frac{\Gamma(m+i)}{\Gamma(1+i)} \frac{z^i}{(z+1)^{m+i}}.
\label{eq:Gamma_CDF}
\end{align}
Note that the ratios of Gamma random variables are known in much more general settings, e.g.,~\cite{ProJ1989} studies the distribution of the ratio of the powers of two, possibly dependent, random variables where both come from Gamma family, but these generalizations are not required in this paper. The form of the distribution function~\eqref{eq:Gamma_CDF} is such that for a given $k_1, m_1$ and $k_2, m_2$, it is not easy to derive conditions on these variables under which the CCDF of one ratio $Z_{k,m}$ dominates that of the other over the whole range of $z$. Therefore, the above result is of little help in providing direct ordering of two random variables $Z_{k_1, m_1}$ and $Z_{k_2, m_2}$. We take an indirect route, which uses the following technical result about the equivalence in distribution of the Gamma random variable and the sum of i.i.d. exponentially distributed random variables. Given this technical result, the main result can then be proved using {\em coupling} arguments. We first state the equivalence result, which is well-known and can be easily verified using characteristic functions. We then remark on how to use coupling arguments to establish stochastic dominance before stating the main result.

\begin{lemma} \label{lem:equivalence}
For i.i.d $\{X_i\}$, with $X_i \sim \exp(1)$, the random variable $X = \sum_{i=1}^k X_i$ is $X \sim \Gamma (k,1)$.
\end{lemma}

\begin{remark}[Using coupling to establish stochastic dominance]
One way to prove $Z_1 \geq_{\rm st} Z_2$ is to find two random variables $Z^*_{1}$ and $Z^*_{2}$ with the same distributions as $Z_{1}$ and $Z_{2}$, respectively, such that it is always the case that $Z^*_{1} \geq Z^*_{2}$. This approach of using the same source of randomness to generate two random variables $Z^*_{1}$ and $Z^*_{2}$ satisfying the above relation and thereby establishing the stochastic dominance result is termed as {\em coupling}~\cite{RosB2011}.
\end{remark}

We now prove the following result on the ordering of the ratios of the Gamma random variables.

\begin{lemma}[Ordering of the ratios of Gamma rvs] \label{lem:gamma_stocorder}
A random variable $Z_{k_1,m_1}$ defined as the ratio of two Gamma random variables  (first order) stochastically dominates $Z_{k_2,m_2}$ if $k_1 \geq k_2$ and $m_1 \leq m_2$.
\end{lemma}

\begin{proof}
Using the equivalence in distribution of the Gamma random variable and the sum of exponential random variables given by Lemma~\ref{lem:equivalence}, we can generate a random variable $Z^*_{k_1,m_1}$ with the same distribution as $Z_{k_1,m_1}$ as follows
\begin{align}
Z^*_{k_1,m_1} = \frac{\sum_{i=1}^{k_1} Y_{1,i}}{\sum_{j=1}^{m_1} Y_{2,i}},
\end{align}
where $\{Y_{m,n}\}$ is the set of i.i.d. random variables such that $Y_{m,n} \sim \exp(1)$. This equivalent representation will facilitate the use of standard coupling arguments, under which the goal now is to generate another random variable $Z^*_{k_2,m_2}$ with the same sources of randomness as that of $Z^*_{k_1,m_1}$, which has the same distribution as $Z_{k_2,m_2}$, and show that $Z^*_{k_1,m_1} \geq Z^*_{k_2,m_2}$. Under the condition $k_1 \geq k_2$, this can be achieved by expressing $Z^*_{k_1,m_1}$ as follows
\begin{align}
Z^*_{k_1,m_1} &= \frac{\sum_{i=1}^{k_2} Y_{1,i} + \sum_{i=k_2+1}^{k_1} Y_{1,i}}{\sum_{j=1}^{m_1} Y_{2,i}}\\
&\geq \frac{\sum_{i=1}^{k_2} Y_{1,i}}{\sum_{j=1}^{m_1} Y_{2,i}} 
\stackrel{(a)}{\geq} \frac{\sum_{i=1}^{k_2} Y_{1,i}}{\sum_{j=1}^{m_2} Y_{2,i}}
\stackrel{(b)}{=} Z^*_{k_2,m_2},
\end{align}
where $(a)$ follows from the condition $m_1 \leq m_2$, and $(b)$ from Lemma~\ref{lem:equivalence}. This completes the proof.
\end{proof}

\begin{figure}[t!]
\centering
\includegraphics[width=\columnwidth]{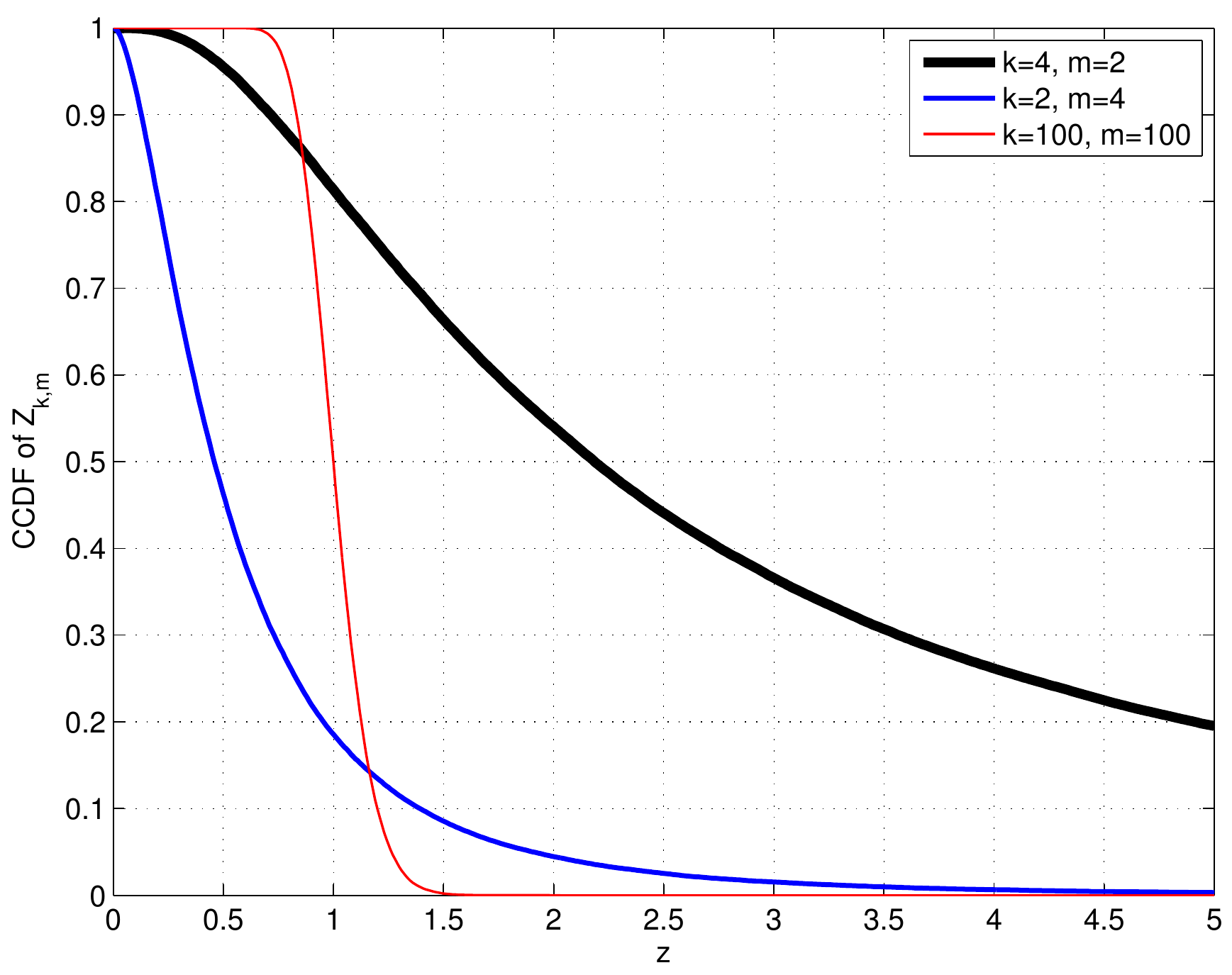}
\caption{The CCDFs of the ratios of Gamma random variables for different shape parameters $k$ and $m$.}
\label{fig:GammaRatio}
\end{figure}

\begin{remark} The set of conditions $k_1 \geq k_2$ and $m_1 \leq m_2$ is stronger than the single condition $k_1/m_1 \geq k_2/m_2$, which first comes to mind from the equivalence of Gamma random variables and the sum of exponential random variables stated in Lemma~\ref{lem:equivalence}. In fact it is easy to argue that the above stochastic dominance is not always true if the only condition on the variables is $k_1/m_1 \geq k_2/m_2$. For instance, consider a case where $k_2 \gg k_1$ and $m_2 \gg m_1$ such that $k_1/m_1 \geq k_2/m_2$. The distribution of $Z_{k_2,m_2}$ is concentrated around its mean and cannot be dominated by the distribution of $Z_{k_1,m_1}$ due to significant difference in their shapes. This is illustrated in Fig.~\ref{fig:GammaRatio}, where the ratio of Gamma random variables with shape parameters $k_1=4$ and $m_1=2$ does not dominate the one with $k_2=100$ and $m_2=100$ due to concentration, although $k_1/m_1 \geq k_2/m_2$ holds.
\end{remark}

To prove the main ordering results of this section, we need to extend the stochastic dominance result of two random variables to multi-variate function of random variables. The result is given in the following Lemma and follows directly from the coupling arguments~\cite{RosB2011}.

\begin{lemma} \label{lem:func_stocorder}
If $X_i \geq_{\rm st} Y_i$ for all $1\leq i \leq n$, then
\begin{align}
\nbbE[g(X_1, X_2, \ldots, X_n)] \geq \nbbE[g(Y_1,Y_2,\ldots,Y_n)],
\end{align}
for all multi-variate functions $g$ that are non-decreasing in each component.
\end{lemma}

\subsection{Coverage Probability Ordering}

Using Lemma~\ref{lem:func_stocorder}, we now derive the following general result on the coverage probability ordering in $K$-tier open access multi-antenna HetNets. As evident from the analysis and remarked later in this section, all the results and insights carry over to closed access networks as well. Recall that the goal of this analysis is to compare or ``order'' the performance of different systems and not to obtain the exact expressions for the performance metrics in any given system.

\begin{theorem} [Coverage ordering in open access networks] \label{thm:Pc_stocorder}
The coverage probability of a $K$-tier open access HetNet with system parameters $\{\Delta_k\}$ and $\{\Psi_k\}$ is higher than or equal to the one with system parameters $\{\Delta'_k\}$ and $\{\Psi'_k\}$ if $\Delta_k \geq \Delta'_k$ and $\Psi_k \leq \Psi'_k$ for $k \in \ncalK$.
\end{theorem}

\begin{proof}
By definition, the coverage probability for open access networks can be expressed as
\begin{align}
\pc &= \E \mathbf{1}\left(\bigcup_{k\in \ncalK} \max_{x_k \in \Phi_k} \frac{P_k h_{kx_k} \|x_k\|^{-\alpha}}{\sum\limits_{k\in \ncalK} \sum\limits_{y\in \Phi_j \setminus x_k} P_j g_{jy} \|y\|^{-\alpha}} > \T_k \right)\\
&= \E \mathbf{1}\left(\bigcup_{k\in \ncalK} \max_{x_k \in \Phi_k} \frac{P_k \|x_k\|^{-\alpha}}{\sum\limits_{k\in \ncalK} \sum\limits_{y\in \Phi_j \setminus x_k} P_j Z_{jk} \|y\|^{-\alpha}} > \T_k \right),
\label{eq:pc_stocproof}
\end{align}
where with slight abuse of notation (dropping the BS location from the subscript),  $Z_{jk} = g_{jy}/h_{kx_k}$ is defined as the ratio of two Gamma random variables corresponding to the $K$-tier HetNet with system parameters $\{\Delta_k\}$ and $\{\Psi_k\}$. Denote the corresponding ratio for the other system setup by $Z'_{jk}$. By Lemma~\ref{lem:gamma_stocorder}, $Z'_{jk} \geq_{\rm st} Z_{jk}$ if $\Psi'_j \geq \Psi_k$ and $\Delta'_j \leq \Delta_k$. Now it is easy to show that the indicator function in \eqref{eq:pc_stocproof}
\begin{align}
g(\{Z_{jk}\}) = \mathbf{1}\left(\bigcup_{k\in \ncalK} \max_{x_k \in \Phi_k} \frac{P_k \|x_k\|^{-\alpha}}{\sum\limits_{k\in \ncalK} \sum\limits_{y\in \Phi_j \setminus x_k} P_j Z_{jk} \|y\|^{-\alpha}} > \T_k \right), \nonumber
\end{align}
is an element wise decreasing function of $Z_{jk}$. Therefore, by Lemma~\ref{lem:func_stocorder} the result follows.
\end{proof}
Using this result, we can make some general comments about the coverage probability in certain realistic deployments. We begin by studying the effect of the number of users served in each tier on coverage probability.

\begin{cor}[Effect of number of users]
For two different $K$-tier open access HetNets, with the same number of antennas in each corresponding tier, the one that serves less users in each tier than the other provides higher coverage due to higher beamforming gain.
\end{cor}
The proof of the above corollary directly follows from the fact that under the same number of antennas for two setups, the one that serves less users in each tier has $\Delta_k \geq \Delta'_k$ and $\Psi_k \leq \Psi'_k$ for each tier, leading to higher coverage. An important extension of the above corollary is the comparison of the SDMA with SU-BF systems when the number of transmit antennas in each corresponding tier are the same. The result is stated as the following corollary.

\begin{cor} [SDMA vs. SU-BF]
For two $K$-tier HetNets, with the same number of antennas in each corresponding tier, one performing SU-BF in each tier and the other performing SDMA, the coverage probability in the SU-BF case will always be higher. 
\end{cor}

Another consequence of this general ordering result is the comparison of SISO with SU-BF and SDMA, with the SDMA case specialized to full SDMA. Recall that in case of full SDMA, $\Delta_k = 1$ and $\Psi_k = M_k$ for all the tiers. The result is given in the following corollary.

\begin{cor} [SU-BF vs. SISO vs. full SDMA]
For three $K$-tier HetNet setups, one performing SU-BF in each tier, another doing SISO transmission in each tier and the last one doing full SDMA in each tier, the coverage probability in case of SU-BF is higher than that of SISO, which in turn outperforms full SDMA. The number of antennas in the corresponding tiers of SU-BF and full SDMA HetNets need not be the same.
\end{cor}

The proof follows from the fact that the shape parameters in case of SU-BF are $\Delta_k = M_k$ and $\Psi_k = 1$, where $M_k>1$ is the number of antennas; in case of SISO are $\Delta'_k = \Psi'_k = 1$; and in case of full SDMA are $\Delta''_k = 1$ and $\Psi''_k = M_k$, where $M_k>1$ is the number of transmit antennas.

For closed access networks, it can be shown that the coverage ordering result derived in Theorem~\ref{thm:Pc_stocorder} holds under a slightly weaker condition because a typical mobile is not allowed to connect to all the tiers. The result is given as the following Corollary of Theorem~\ref{thm:Pc_stocorder} and the proof is skipped. Due to the similarity of this result with the open access case, the insights gained for the open access networks above carry over to the closed access networks as well.

\begin{cor}[Coverage ordering in closed access]
For two HetNets, with $\ncalB \subset \ncalK$ open access tiers, the one with  system parameters $\{\Delta_k\}$ and $\{\Psi_k\}$ has a higher or equal coverage than the one with system parameters $\{\Delta'_k\}$ and $\{\Psi'_k\}$ if $\Delta_k \geq \Delta'_k$ for $k \in \ncalB$ and $\Psi_k \leq \Psi'_k$ for $k \in \ncalK$.
\end{cor}

\subsection{Ordering Result for Rate per User}
Another metric of interest for the performance evaluation of HetNets is the rate achievable per user. In addition to the link quality (characterized in terms of $\sir$), it also depends upon the effective resources allocated to each user. For tractability, we make following two assumptions on resource allocation: i) each $k^{th}$ tier BS serves same number of users, and ii) each BS allocates equal time-frequency resources to all its users. For SDMA, it should be noted that several users will be scheduled on the same time-frequency resource block. Interested readers can refer to~\cite{SinDhiJ2013} for more details on the motivation and validation of these assumptions. Under these assumptions, we denote the effective time-frequency resources, e.g., bandwidth, allocated to a user served by a $k^{th}$ tier BS by $\ncalO_k$. The two assumptions on resource allocation ensure that $\ncalO_k$ is the same for all the users served by any $k^{th}$ tier BS. Therefore, the downlink rate of a typical user served by a $k^{th}$ tier BS located at $x_k \in \Phi_k$ is
\begin{align}
R(x_k) = \ncalO_k \log_2(1+\sir(x_k)).
\end{align}
Due to the difficulties in modeling exact load on each BS~\cite{DhiGanJ2012}, which often requires characterization of the service areas for different types of BSs, it is challenging to characterize $\ncalO_k$ and hence derive exact expressions for per user rate distribution~\cite{SinDhiJ2013}. However, we now show that to compare different multi-antenna transmission techniques in certain cases, this characterization is not required and the general ordering result derived above for the coverage probability can be easily extended for the rate per user as well. Before going into the technical details, it is important to note that the loading across tiers may differ significantly due to the orders of magnitude differences in their coverage footprints. Therefore, the effective resources $\ncalO_k$ available in small cells for each user might be significantly higher than the macrocells. In such a case, it might be beneficial for a user to connect to a small cell even though it may not provide the best $\sir$ over the network. We will capture this characteristic of HetNets in our definition of rate distribution below. Due to the interpretation of a minimum rate required by each application, e.g., video, we will study rate distribution in terms of ``rate coverage'', which is defined below. It is just the CCDF of rate when the target rates are the same for all the tiers.

\begin{ndef}[Rate Coverage]
A typical user in an open access network is said to be in rate coverage if its effective downlink rate from at least one of the BSs in the network is higher than the corresponding target. We denote the target rate for a $k^{th}$ tier BS as $T_k$. Rate coverage can now be mathematically expressed as
\begin{align}
\rc = \nbbP\left(\bigcup_{k \in \ncalK} \max_{x \in \Phi_k} \ncalO_k \log_2(1+\sir(x_k)) > T_k  \right).
\end{align}
For closed access, the expression remains the same except that the union is now over the set $\ncalB \subset \ncalK$.
\end{ndef}

It is easy to show that the rate coverage for open access networks is always higher than the closed access networks. This follows from the same arguments that were used in case of coverage probability earlier in this section. We now state the main ordering result for rate coverage in the following Theorem.

\begin{theorem}[Ordering result for rate coverage] \label{thm:PcR_stocorder}
For two $K$-tier HetNets with the same resource allocation per user $\ncalO_k$ for each corresponding tier, the one with system parameters $\{\Delta_k\}$ and $\{\Psi_k\}$ has equal or higher rate coverage than the one with system parameters $\{\Delta'_k\}$ and $\{\Psi'_k\}$ if $\Psi_k \leq \Psi'_k$ for all $k \in \ncalK$, and $\Delta_k \geq \Delta'_k$ for all $k \in \ncalK$ in open access and all $k \in \ncalB$ in closed access.
\end{theorem}

\begin{proof}
The rate coverage can be expressed as
\begin{align}
\rc = \nbbE \nb1 \left(\bigcup_{k \in \ncalK} \max_{x \in \Phi_k} \ncalO_k \log_2(1+\sir(x_k)) > T_k  \right),
\label{eq:pcR_Exp}
\end{align}
where $\sir(x_k)$ can be expressed as
\begin{align}
\sir(x_k) &= \frac{P_k h_{kx_k} \|x_k\|^{-\alpha}}{\sum\limits_{k\in \ncalK} \sum\limits_{y\in \Phi_j \setminus x_k} P_j g_{jy} \|y\|^{-\alpha}},\\
&= \frac{P_k \|x_k\|^{-\alpha}}{\sum\limits_{k\in \ncalK} \sum\limits_{y\in \Phi_j \setminus x_k} P_j Z_{jk} \|y\|^{-\alpha}},
\end{align}
where as in the proof of Theorem~\ref{thm:Pc_stocorder}, we define $Z_{jk} = g_{jy}/h_{k x_k}$ as the ratio of the two Gamma random variables. Now note that the indicator function inside the expectation of \eqref{eq:pcR_Exp} is an element wise decreasing function of $Z_{jk}$, from which the result follows on the same lines as the proof of Theorem~\ref{thm:Pc_stocorder}.
\end{proof}

\begin{remark}[Same ordering for coverage and rate per user]
From Theorems~\ref{thm:Pc_stocorder} and~\ref{thm:PcR_stocorder}, we note that the ordering conditions for rate per user in the above setup are the same as that of coverage probability. Therefore, all the conclusions, including the ordering of SDMA, SU-BF and SISO transmission techniques, derived for coverage probability carry over to the rate per user case as well.
\end{remark}

Although it is clear from the above discussion that both SU-BF and SISO outperform SDMA, both in terms of coverage probability and average rate per user, it is important to note that we have not yet accounted for the fact that SDMA serves more users than both SISO and SU-BF, and may result in a higher sum-data rate. To address this, we compare the three transmission techniques in terms of ASE in the next section.

\section{Coverage Probability and ASE Performance}
This is the second main technical section of the paper where we derive an upper bound on the coverage probability of a typical user in a $K$-tier HetNet, where the transmission techniques adopted by each tier are characterized in terms of the shape parameters $\Delta_k$ and $\Psi_k$ of the Gamma rvs. Recall that the coverage probability is formally defined in Definition~\ref{def:pc}. For this analysis, we assume that BS locations of each tier are drawn from an independent PPP $\Phi_k$ with density $\lambda_k$. Although this model is not as general as the one considered in the previous section, it is likely accurate in modeling the opportunistic deployment of small cells and has been validated for planned tiers, such as single-antenna macrocells by empirical observations~\cite{TayDhiC2012} and theoretical arguments under sufficient channel randomness~\cite{BlaKarJ2012}. In this paper, we validate it in the context of coverage probability in MIMO HetNets by comparing it with an actual $4$G deployment and the popular grid model in the numerical results section.

\subsection{Upper Bound on Coverage Probability}
Before deriving the upper bound, we first derive an expression for the Laplace transform of interference. The result is given in Lemma~\ref{lem:laplace} and the proof is given in the Appendix~\ref{appendix:Laplace}. This generalizes the Laplace transform of interference derived for $K$-tier SISO HetNets with Rayleigh fading, i.e., exponential channel powers, in Theorem 1 of~\cite{DhiGanJ2012}.

\begin{lemma} \label{lem:laplace}
The Laplace transform of interference $\mathcal{L}_I(s) = \E\left[e^{-sI}\right]$, where $I = \sum_{k\in \ncalK} \sum_{y\in \Phi_j} P_j g_{jy} \|y\|^{-\alpha}$ is
\begin{align}
\mathcal{L}_I(s) = \exp\left( - s^{\frac{2}{\alpha}} \sum_{j \in \ncalK} \lambda_j P_j^{\frac{2}{\alpha}}   C (\alpha, \Psi_j) \right),
\end{align}
where
\begin{align}
C (\alpha, \Psi_j) = \frac{2 \pi }{\alpha} \sum_{m=1}^{\Psi_j} {\Psi_j \choose m}  B\left(\Psi_j -m + \frac{2}{\alpha}, m - \frac{2}{\alpha}\right),
\label{eq:CalphaM}
\end{align}
and $B(x,y) = \int_{0}^1 t^{x-1} (1-t)^{y-1} \nrmd  t$ is Euler's Beta function.
\end{lemma}
Using this result, we now derive an upper bound on the coverage probability and the result is given in Theorem~\ref{thm:MIMOmain}. 

\begin{theorem} \label{thm:MIMOmain}
The coverage probability of a typical user in a $K$-tier open access HetNet is upper bounded by
\begin{align}
\label{eq:MIMOmain}
\pc & \leq \sum_{k \in \ncalK} \lambda_k \ncalA_k,
\end{align}
where $s_{x_k} = \T_k \|x_k\|^\alpha P_k^{-1}$ and
\begin{align}
\ncalA_k = \sum_{i=0}^{\Delta_k-1} \frac{1}{i!} \int\limits_{x_k \in \R^2} (-s_{x_k})^i \frac{\delta^i}{\delta (s_{x_k})^i } \mathcal{L}_{I_{x_k}}(s_{x_k}) \nrmd  x_k.
\end{align}
The upper bound for closed access networks is the same except that the summation in~\eqref{eq:MIMOmain} is over the set $\ncalB$ instead of $\ncalK$.
\end{theorem}


\begin{proof}
We prove the result for open access networks and will highlight exactly where the proof will differ for closed access networks. Starting with the definition of the coverage probability, we have 
\begin{align}
\pc &= \E \left[1 \left(\bigcup_{k \in \ncalK} \bigcup_{x_k \in \Phi_k} 
\sir(x_k) > \T_k \right)\right]\\
&\stackrel{(a)}{\leq} \E \left[\sum_{k \in \ncalK} \sum_{x_k \in \Phi_k} 1 \left( \sir(x_k) > \T_k \right)\right] \label{eq:unionbound}\\
&= \sum_{k \in \ncalK}\E \left[ \sum_{x_k \in \Phi_k} 1 \left( P_k h_{kx_k} \|x_k\|^{-\alpha}> \T_k I_{x_k} \right)\right], \label{eq:intermed1}
\end{align}
where $(a)$ follows from the union bound and $I_{x_k}$ is the interference received by the typical user when it is connected to the $k^{th}$ tier BS located at $x_k$, i.e.,
\begin{align}
I_{x_k} = \sum_{j \in \ncalK} \sum_{y \in \Phi_j \setminus x_k} P_j g_{jy} \|y\|^{-\alpha}.
\end{align}
Note that for closed access, the summation in \eqref{eq:unionbound} and \eqref{eq:intermed1} will be over $\ncalB$ instead of $\ncalK$. This is the only difference in the proofs of open and closed access cases. Continuing with the proof of open access case, since the channel power of the direct link is independent of everything else, we can take the expectation w.r.t. $h_{kx_k}$ inside~\eqref{eq:intermed1} to write the coverage probability as
\begin{align}
\pc &\leq \sum_{k \in \ncalK} \E \left[\sum_{x_k \in \Phi_k} \P \left( h_{kx_k} > \T_k I_{x_k} \|x_k\|^{\alpha} P_k^{-1} \right)\right].
\label{eq:intermed2}
\end{align}
Now we first evaluate the probability $\P(h_{kx} > z)$ as
\begin{align}
\P(h_{kx} > z) &\stackrel{(a)}{=} \frac{\Gamma(\Delta_k, z)}{\Gamma(\Delta_k)}
\stackrel{(b)}{=} e^{-z} \sum_{i=0}^{\Delta_k-1} \frac{z^i}{i!},
\label{eq:intermed3}
\end{align}
where $(a)$ follows from $h_{kx} \sim \Gamma(\Delta_k,1)$, and $\Gamma(\Delta_k, z)$ in the numerator is the upper incomplete Gamma function given by $\Gamma(\Delta_k, z) = \int_{z}^{\infty} u^{\Delta_k - 1} e^{-u} \nrmd  u$, $(b)$ follows by specializing the expression of incomplete Gamma function for the case when $\Delta_k$ is an integer. Now denote $\beta_k \|x_k\|^{\alpha} P_k^{-1}$ by $s_{x_k}$ and substitute \eqref{eq:intermed3} in \eqref{eq:intermed2} to get
\begin{align}
\pc &\leq \sum_{k \in \ncalK} \E \sum_{x_k \in \Phi_k} e^{-s_{x_k} I_{x_k}} \sum_{i=0}^{\Delta_k-1} \frac{(s_{x_k} I_{x_k})^i}{i!}\\
&\stackrel{(a)}{=} \sum_{k \in \ncalK}  \lambda_k \int\limits_{x_k \in \R^2} \E_{I_{x_k}} e^{-s_{x_k} I_{x_k}} \sum_{i=0}^{\Delta_k-1} \frac{(s_{x_k} I_{x_k})^i}{i!} \nrmd  x_k\\
&= \sum_{k \in \ncalK}  \lambda_k \sum_{i=0}^{\Delta_k-1} \frac{1}{i!} \int\limits_{x_k \in \R^2} \E_{I_{x_k}}e^{-s_{x_k} I_{x_k}} (s_{x_k} I_{x_k})^i \nrmd  x_k,
\end{align}
where $(a)$ follows from Campbell-Mecke Theorem~\cite{StoKenB1995}. Now note that if $\Delta_k$ were $1$, the expectation term is just $\mathcal{L}_{I_{x_k}}(s_{x_k})$, i.e., the Laplace transform of interference evaluated at $s_{x_k}$. For $\Delta_k > 1$, we evaluate the expectation in terms of the derivative of the Laplace transform as
\begin{align}
\E_{I_{x_k}}\left[ e^{-s_{x_k} I_{x_k}} (s_{x_k} I_{x_k})^i \right] &\stackrel{(a)}{=} s_{x_k}^i \mathcal{L} \{t^i f_{I_{x_k}}(t) \}(s_{x_k})\\
&\stackrel{(b)}{=} (-s_{x_k})^i \frac{\delta^i}{\delta (s_{x_k})^i } \mathcal{L}_{I_{x_k}}(s_{x_k}),
\label{eq:intermed4}
\end{align}
where $(a)$ follows from the definition of the Laplace transform and $(b)$ follows from the identity $t^n f(t) \longleftrightarrow (-1)^n \frac{\delta^n}{\delta (s)^n} \mathcal{L}\{f(t)\} (s)$. Substituting this in \eqref{eq:intermed4}, we can express the upper bound on coverage probability in terms of Laplace transform of interference as
\begin{align}
\pc & \leq \sum_{k \in \ncalK} \lambda_k \sum_{i=0}^{\Delta_k-1} \frac{1}{i!} \int\limits_{x_k \in \R^2} (-s_{x_k})^i \frac{\delta^i}{\delta (s_{x_k})^i } \mathcal{L}_{I_{x_k}}(s_{x_k}) \nrmd  x_k,
\label{eq:intermed5}
\end{align}
which completes the proof.
\end{proof}
We note that the above upper bound involves a derivative of Laplace transform, which makes its numerical evaluation difficult. However, it is possible to reduce the upper bound to a simple closed form for full SDMA and easy to evaluate numerical expressions in the other cases. The simplified result is given in the following Corollary.

\begin{cor}
For $\Delta_k = 1$, $\ncalA_k$ can be reduced to
\begin{align}
\ncalA_k = \frac{\pi P_k^{\frac{2}{\alpha}} \T_k^{\frac{2}{\alpha}}}{\sum_{j \in \ncalK} \lambda_j P_j^{\frac{2}{\alpha}} C(\alpha, M_j)},
\end{align}
and for $\Delta_k > 1$ to
\begin{align}
\ncalA_k =&\ \sum_{i=0}^{\Delta_k-1} \frac{1}{i!} \sum \frac{i!}{j_{1}! j_{2}! \ldots j_{i}!} \int\limits_{x_k \in \R^2} (-s_{x_k})^i e^{-\mathcal{C} s_{x_k}^{\frac{2}{\alpha}}} \nonumber \\
&\ \prod_{\ell = 1}^{i} \frac{1}{(\ell!)^{jl}}\left(-\mathcal{C} s_{x_k}^{\frac{2}{\alpha}-\ell} \displaystyle \prod_{n=0}^{\ell-1}\left(\frac{2}{\alpha}-n\right)\right)^{j_{\ell}} \nrmd  x_k.
\end{align}
\end{cor}
\begin{proof}
For $\Delta_k = 1$,
\begin{align}
\ncalA_k &= \int_{x_k \in \nbbR^2} \ncalL_{I_{x_k}} \left(\T_k \|x_k\|^\alpha P_k^{-1} \right) \nrmd x_k \\
&\stackrel{(a)}{=} \int\limits_{x_k \in \nbbR^2} \exp \left(- \T_k^{\frac{2}{\alpha}} \|x_k\|^{2} P_k^{-\frac{2}{\alpha}} \ncalC \right) \nrmd x_k,
\end{align}
where $(a)$ follows from the Laplace transform expression derived in Lemma~\ref{lem:laplace}. Recall that $\ncalC = \sum_{j \in \ncalK} \lambda_j P_j^{\frac{2}{\alpha}}   C (\alpha, \Psi_j)$. The closed form expression now follows directly by converting the integral from Cartesian to polar coordinates. 

For $\Delta_k > 1$, using the Laplace transform expression and calculating its derivative using Fa\`{a} di Bruno's formula for the composite function $(f \circ g)(s_{x_k})$, with $f(s_{x_k}) = \exp \left(s_{x_k}\right)$, and $g(s_{x_k}) = -\mathcal{C} s_{x_k}^{\frac{2}{\alpha}}$, the result follows.
\end{proof}

We note that the upper bound is in closed form if $\Delta_k = 1$ for all tiers. The result is given in the following corollary. Even for $\Delta_k > 1$, the upper bound can be numerically computed fairly easily, especially for small values of $\Delta_k$.

\begin{cor} \label{thm:SDMAmain}
The coverage probability in a $K$-tier open access HetNet with each $k^{th}$ tier BS performing full SDMA to serve $M_k$ users, i.e., $\Delta_k = 1$ $\forall\ k\in \ncalK$, is given by
\begin{align}
\pc \leq \pi  \frac{\sum_{k \in \ncalK} \lambda_k P_k^{\frac{2}{\alpha}} \beta_k^{-\frac{2}{\alpha}}}{\sum_{j =1}^K \lambda_j P_j^{\frac{2}{\alpha}} C(\alpha, M_j)}.
\end{align}
For the closed access case, the summation in the numerator is over $\ncalB$ instead of $\ncalK$.
\end{cor}

We now comment on the tightness of the coverage probability upper bound in various regimes and for various transmission techniques.

\subsection{Tightness of the Upper Bound}
For conciseness, we will focus on the open access networks, with the understanding that all the arguments remain the same for closed access case.  Since the bound is derived by using the union bound in \eqref{eq:unionbound}, the tightness depends upon the number of candidate BSs that provide $\sir$ greater than the target $\sir$. Denote this random variable by $X(\{\Delta_k\}, \{\Psi_k\})$, which can be expressed as
\begin{align}
X(\{\Delta_k\}, \{\Psi_k\}) = \sum_{k \in \ncalK} \sum_{x_k \in \Phi_k} 1 \left( \sir(x_k) > \T_k \right).
\end{align}
The bound holds with equality if there is strictly one candidate serving BS for a typical user, i.e., $\nbbP(X(\{\Delta_k\}, \{\Psi_k\}) > 1) = 0$. This is the case in SISO HetNets for $\beta_k > 1,\ \forall k,$ as shown in~\cite{DhiGanJ2012}. For any other general system configuration, the tightness of the bound depends upon whether the probability $\nbbP(X(\{\Delta_k\}, \{\Psi_k\}) > 1)$ is close to zero or not. In general, it is hard to evaluate simple expressions for this probability. However, it is possible to make a few simple observations about the expected tightness of the bound. For instance, the bound gets tight with the increasing values of target $\sir$s because $X(\{\Delta_k\}, \{\Psi_k\})$ is an element-wise decreasing function of $\T_k$. For further insights, we derive the following ordering result for $X(\{\Delta_k\}, \{\Psi_k\})$. The proof follows using Lemma~\ref{lem:gamma_stocorder} on the same lines as that of Theorem~\ref{thm:Pc_stocorder}.
 
\begin{theorem}[Ordering result for $X$] \label{thm:X_stocorder}
If $\Delta_k \geq \Delta'_k$ and $\Psi_k \leq \Psi'_k$ $\forall\ k$, then $X(\{\Delta_k\}, \{\Psi_k\})$ (first order) stochastically dominates $X(\{\Delta'_k\}, \{\Psi'_k\})$, i.e. $\nbbP( X(\{\Delta_k\}, \{\Psi_k\}) > n) \geq \nbbP( X(\{\Delta'_k\}, \{\Psi'_k\}) > n)$, $\forall\ n$. 
\end{theorem}

\begin{proof}
Using the alternate expression of $\sir$ given by \eqref{eq:SIR_alternate}, express $X(\{\Delta_k\}, \{\Psi_k\})$ in terms of the channel power gains as
\begin{align} 
\sum_{k \in \ncalK} \sum_{x_k \in \Phi_k} 1 \left( \frac{P_k \|x_k\|^{-\alpha}}{\sum_{k\in \ncalK} \sum_{y\in \Phi_j \setminus x_k} P_j Z_{jk} \|y\|^{-\alpha}} > \T_k \right),
\end{align}
where $Z_{jk} = g_{jy}/h_{kx_k}$ is the ratio of the two Gamma random variables. For another system with $Z'_{jk} = g'_{jy}/h'_{kx_k}$, $Z_{jk} \leq_{\rm st} Z'_{jk}$ if $\Delta_k \geq \Delta'_k$ and $\Psi_j \leq \Psi'_j$, which follows from Lemma~\ref{lem:gamma_stocorder}. The result now follows on the same lines as the proof of Theorem~\ref{thm:Pc_stocorder} using Lemma~\ref{lem:func_stocorder} along with the fact that $X(\{\Delta_k\}, \{\Psi_k\})$ is an element-wise non-increasing function of $Z_{jk}$.
\end{proof}

\begin{remark}[Tight Bound in case of SDMA] \label{rem:tight}
One of the useful consequences of Theorem~\ref{thm:X_stocorder} is the prediction of the tightness of the upper bound for SDMA. One interpretation of the above result is that the bound gets tighter when all the BSs serve more users, i.e., $\Delta_k$ decreases and $\Psi_k$ increases for all the tiers. A limiting case is that of full SDMA, where the number of users served by each BS is equal to the number of its transmit antennas. Beyond this point, the bound gets tighter with the addition of more transmit antennas keeping $\Delta_k=1$. We revisit these observations in the numerical results section and show that the bound is in fact surprisingly tight even for two transmit antennas down to very low target $\sir$s.
\end{remark}

In the rest of this section, we will mainly focus on the full SDMA case. Recall that in this case $\Delta_k = 1$ and $\Psi_j = M_j$ and the coverage probability upper bound is given by Corollary~\ref{thm:SDMAmain}. As argued in Remark~\ref{rem:tight} and validated in the numerical results section, the closed form upper bound is tight and can be used as an approximation for the coverage probability. For simplicity we will use equality instead of an approximation.

\begin{remark}[Similarity with $\pc$ in SISO case]
The coverage probability expression derived for full SDMA case in Corollary \ref{thm:SDMAmain} has a striking similarity with the coverage probability in the SISO case derived in~\cite{DhiGanJ2012}. The only difference is that the constant $C(\alpha, M_j)$ in that case is simply $C(\alpha) = \frac{2 \pi^2 \csc\left(\frac{2 \pi}{\alpha} \right)}{\alpha}$.
\end{remark}

To facilitate direct comparison of the full SDMA and the SISO cases, we need to understand the relationship between $C(\alpha)$ and $C(\alpha,M)$. Let us take a closer look at the expression of $C(\alpha, M)$ given by \eqref{eq:CalphaM}.
First note that $C(\alpha, M)$ is an increasing function of $M$. Now let us evaluate $C(\alpha, 1)$:
\begin{align}
C(\alpha,1) &= \frac{2 \pi}{\alpha} B\left(\frac{2}{\alpha}, 1- \frac{2}{\alpha}\right)\\
& = \frac{2 \pi}{\alpha} \Gamma\left(\frac{2}{\alpha}\right) \Gamma\left(1 - \frac{2}{\alpha}\right) = \frac{2 \pi^2 \csc\left(\frac{2 \pi}{\alpha} \right)}{\alpha},
\end{align}
where the last step follows by Euler's reflection formula. Hence $C(\alpha,1)$ is the same as $C(\alpha)$ derived for the SISO case in~\cite{DhiGanJ2012}. From the monotonicity of $C(\alpha,M)$ it follows that $C(\alpha,M) > C(\alpha)$ $\forall M>1$.

\begin{remark}[Full SDMA vs. SISO coverage]
Keeping all the system parameters the same, the full SDMA coverage is always lower than that of the SISO case. This is consistent with the coverage probability ordering results derived in the previous section.
\end{remark}

\begin{remark}[Scale invariance in open access HetNets] \label{rem:scale_invariant}
The full SDMA coverage probability is invariant to the density of the BSs, number of tiers and the transmit powers when the target $\sir$s and the number of transmit antennas are the same for all the tiers in open access HetNets. The coverage probability in this case is given by $\pc = \frac{\pi }{C(\alpha,M)}\T^{-\frac{2}{\alpha}}$. This result is again similar to the SISO result where the coverage probability reduces to $\pc = \frac{\pi }{C(\alpha)}\T^{-\frac{2}{\alpha}}$. The scale invariance result does not hold for closed access HetNets.
\end{remark}

\subsection{Area Spectral Efficiency}
Although the comparison of various system configurations and transmission techniques is quite conclusive in terms of coverage probability and the rate per user, it does not directly account for the fact that some techniques, such as SDMA, serve higher number of users than the others, such as SU-BF, and may result in higher sum data rate. To account for this fact, we consider ASE, which gives the number of bits transmitted per unit area per unit time per unit bandwidth. For a multi-tier setup, it can be formally defined as 
\begin{align}
\eta = \sum_{k \in \ncalK} \Psi_k \lambda_k \log_2(1+\T_k) \pc^{(k)},
\end{align}
where $\pc^{(k)}$ is the per tier coverage probability, i.e., coverage probability conditional on the serving BS being in the $k^{th}$ tier. Since the derivations of per tier coverage probabilities are out of the scope of this paper, for analytical comparisons we limit our discussion to the cases where $\pc^{(k)} = \pc$ for all tiers. This is guaranteed for the SISO case when the target $\sir$s are the same for all tiers and for SDMA when additionally the number of antennas per BS are also the same for all tiers. Recall that the coverage probabilities under these assumptions are scale invariant, as discussed in Remark~\ref{rem:scale_invariant}. 
The ASE under these assumptions can be expressed as
\begin{align}
\eta =  \pc \log_2(1+\T) \sum_{k \in \ncalK} \Psi_k \lambda_k.
\end{align}
We first compare the ASE of the full SDMA and the SISO cases below. The ASE for full SDMA case is
\begin{align}
\eta_{M} = M \frac{\pi }{C(\alpha,M)}\T^{-\frac{2}{\alpha}} \log_2(1+\T) \sum_{k \in K}\lambda_k,
\end{align}
and for the SISO case is
\begin{align}
\eta_{S} = \frac{\pi }{C(\alpha)}\T^{-\frac{2}{\alpha}} \log_2(1+\T) \sum_{k \in K}\lambda_k.
\end{align}
The ratio of the ASEs can be expressed as
\begin{align}
\frac{\eta_M}{\eta_S} = \frac{M C(\alpha)}{C(\alpha,M)}.
\label{eq:ASE_ratio}
\end{align}
Using the fact that
\begin{align} 
\lim_{M \to \infty} \frac{C(\alpha,M)}{M^{\frac{2}{\alpha}}} = \pi \Gamma(1-2/\alpha),
\end{align} 
the ratio of the ASEs can be approximated as
\begin{align}
\frac{\eta_M}{\eta_S} &\approx \frac{M^{1-\frac{2}{\alpha}} C(\alpha)}{\pi\Gamma(1-2/\alpha)} \label{eq:ASE_ratioapprox}
=\Gamma\left(1+\frac{2}{\alpha}\right) M^{1-\frac{2}{\alpha}},
\end{align}
which shows that the ratio grows with the number of antennas when $\alpha>2$. In the next section we will validate this observation and show that the ASE in case of full SDMA is always higher than the SISO case. As shown in~\cite{DhiKouC2012}, the approximation is surprisingly tight even for small $M$. 

Another relevant comparison is that of full SDMA and SISO when both the systems are serving the same density of users. To facilitate this comparison, the densities of BSs for SDMA case will be lower than the SISO case by a factor of M. This comparison will provide insights into whether it is beneficial in terms of ASE to deploy $\lambda$ BSs per unit area with $M$ antennas or $M \lambda$ single-antenna BSs per unit area. In that case, the ratio of the ASEs can be approximated as
\begin{align}
\frac{\eta_M}{\acute{\eta_S}} \approx \Gamma\left(1+\frac{2}{\alpha}\right) M^{-\frac{2}{\alpha}}, \label{eq:ASE_ratioapprox2}
\end{align}
which shows that the ratio decreases sublinearly with the number of antennas when $\alpha>2$. We will validate this observation in the next section and show that the ASE in SISO case is higher than that of the full SDMA case.

So far, we have focused only on the comparison between full SDMA and SISO cases, mainly because the coverage probability expressions for these cases are known in closed form. Since the coverage probability upper bound for SU-BF cannot be reduced to closed form and moreover the tightness of the bound is questionable, we cannot perform similar comparisons with SU-BF unless we derive a simple coverage probability expression, which is out of the scope of this paper. That being said, it is possible to compare the three cases in the very low target $\sir$ regime. Note that this case is of practical relevance since current wireless standards support communication down to very low $\sir$s, which is about $-6$ dB for 3GPP LTE~\cite{3GPP2010b}. 
\begin{prop}[ASE comparison for vanishingly small $\sir$ targets] \label{prop:ASE}
For the same infrastructure, i.e., the densities of BSs, the ASEs of SU-BF and SISO are the same and of full SDMA is higher than the both when $\T_k \rightarrow 0$ for all $k$.
\end{prop}
\begin{proof}
The proof follows from the fact that coverage probability is an element wise decreasing function of $\{\T_k\}$ and approaches $1$ when $\T_k \rightarrow 0$ for all $k$. Therefore, the ASE for this regime is
\begin{align}
\eta = \sum_{k \in \ncalK} \Psi_k \lambda_k \log_2(1+\T_k),
\end{align}
from which the result follows by the fact that $\Psi_k = M_k > 1$ for all $k$ for full SDMA and $\Psi_k = 1$ for all $k$ for both SU-BF and SISO cases.
\end{proof}
We will revisit this result along with the ASE comparison in the moderate and high target $\sir$ regimes in the next section.

\section{Numerical Results}

\begin{figure}[t]
\centering
\includegraphics[width=\columnwidth]{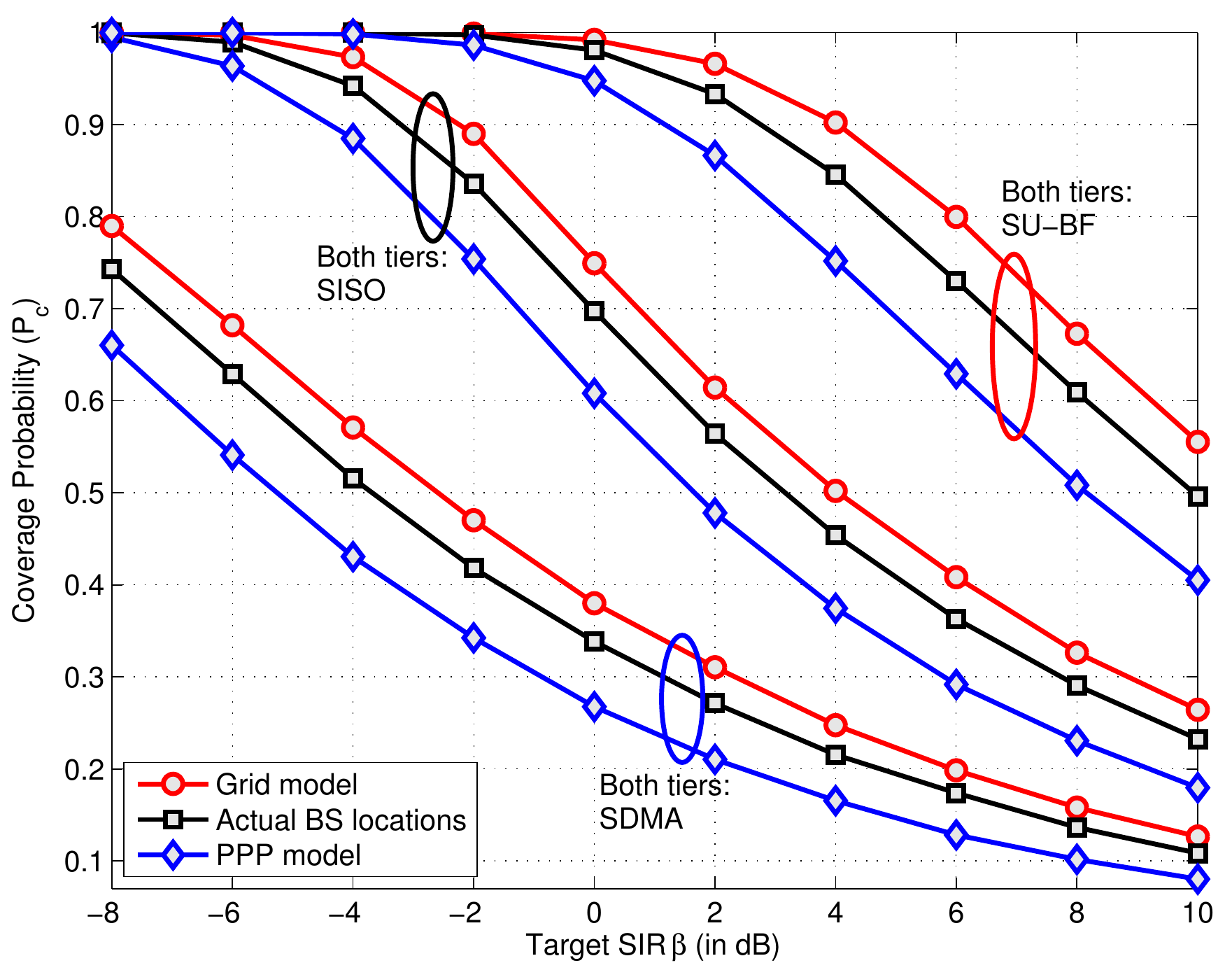}
\caption{Coverage probability for three different models for macrocells. The second tier is PPP in all the cases. ($K=2, P = [1, .01], \lambda_2 = 2\lambda_1, \alpha = 3.8$).  The number of antennas in case of multi-antenna tiers is $M = 4$.}
\label{fig:Grid_Actual_PPP}
\end{figure}

\begin{figure}[t]
\centering
\includegraphics[width=\columnwidth]{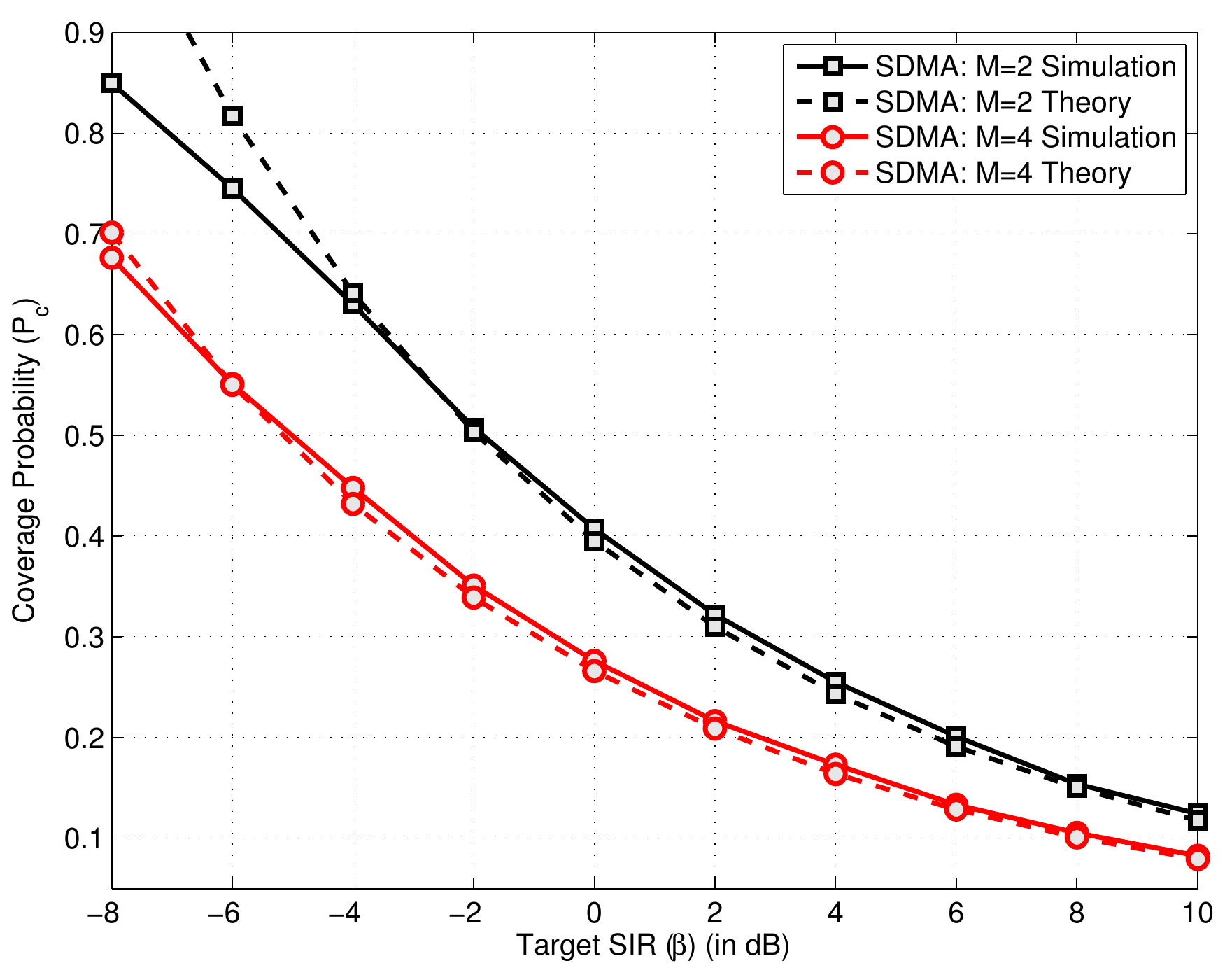}
\caption{Coverage probability of a two-tier HetNet when both tiers perform full SDMA ($K=2, P = [1, .01], M_1 = M_2 = M, \lambda_2 = 2\lambda_1, \beta_1 = \beta_2, \alpha = 3.8$).}
\label{fig:SDMA_bound}
\end{figure}


Since there is a slight difference in the simulation of the proposed multi-antenna model and the ones proposed in the literature for the SISO HetNets, e.g.,~\cite{DhiGanJ2012}, we will briefly summarize the simulation procedure before explaining the results. Choose a sufficiently large window and simulate the locations of different classes of BSs as realizations of independent PPPs of given densities. Associate two independent marks $h_x$ and $g_x$ with each BS. Assuming the typical user lies at the origin, calculate the desired signal strength from each BS using the sequence of marks $\{h_{kx}\}$ and the interference power using the sequence $\{g_{kx}\}$. Calculate the received $\sir$ from each BS. The user is now said to be in coverage if the received $\sir$ from at least one of the BSs belonging to the permissible tiers is more than the corresponding target. Repeating this procedure sufficient number of times, we have an estimate of the coverage probability. Using this procedure, we first validate the location model and establish the tightness of the upper bound for SDMA in the following subsection. Note that since we are focusing on the interference limited regime in this paper, the absolute values of transmit powers and deployment densities are irrelevant. The results only depend on their respective ratios.

\subsection{Model validation and tightness of the upper bound on $\pc$}
Recall that while the PPP model is sensible for small cells, especially the ones deployed without planning, such as femtocells, it is dubious for centrally planned tiers, such as macrocells. Therefore, to validate the proposed location model for MIMO HetNets, we consider following three setups for a two-tier HetNet with a special focus on macrocells: i) the macrocells are modeled by a hexagonal grid, ii) the macrocell locations are drawn from an actual $4$G deployment over $40 \times 40$ km${^2}$ area~\cite{AndBacJ2011,DhiGanJ2012}, iii) the macrocell locations are drawn from an independent PPP, as in the proposed model. The second tier is modeled as a PPP in all three cases. Note that the actual BS locations used in this comparison can be accurately modeled as a Strauss process, as shown in~\cite{TayDhiC2012}. For each of these three location models, we further consider three setups: i) both tiers have 4 transmit antennas per BS and perform SU-BF, i.e. $M_k = 4, \Psi_k=1$ for all $k$, ii) both tiers have a single transmit antenna per BS and perform SISO transmission, i.e. $M_k = 1, \Psi_k=1$ for all $k$, and iii) both tiers have 4 transmit antennas per BS and perform full SDMA, i.e., $M_k = 4, \Psi_k=4$ for all $k$. The simulation procedure remains the same as described above for the PPP model, except of course that the macrocell locations are appropriately drawn from either PPP, grid or actual location data for each setup. From the numerical results presented in Fig.~\ref{fig:Grid_Actual_PPP}, we note that in all three setups, the proposed model provides a lower bound on the coverage probability of an actual $4$G deployment and is about as accurate as the grid model, which provides an upper bound. These observations are consistent with those of \cite{DhiGanJ2012} for SISO HetNets. In the rest of this section, we will focus solely on the proposed model, i.e., each tier is modeled as an independent PPP.

After validating the location model, we numerically evaluate the coverage probability of a two-tier HetNet in full SDMA case and compare the results with the upper bound derived in Corollary~\ref{thm:SDMAmain} in Fig.~\ref{fig:SDMA_bound}. As stated in Remark~\ref{rem:tight}, the bound is tight down to very low target $\sir$s. Even for $M=2$, the bound is tight down to about $-4$ dB. A slight gap, although negligible, at moderate to high target $\sir$s is due to the border effects in simulation, also observed earlier in~\cite{DhiGanJ2012}. In particular, the simulation is performed over a finite window whereas the analysis assumes BSs over an infinite plane. Nevertheless, this validates our assumption of considering the upper bound as an approximation of the coverage probability in case of full SDMA in the previous section.


\begin{figure}[t]
\centering
\includegraphics[width=\columnwidth]{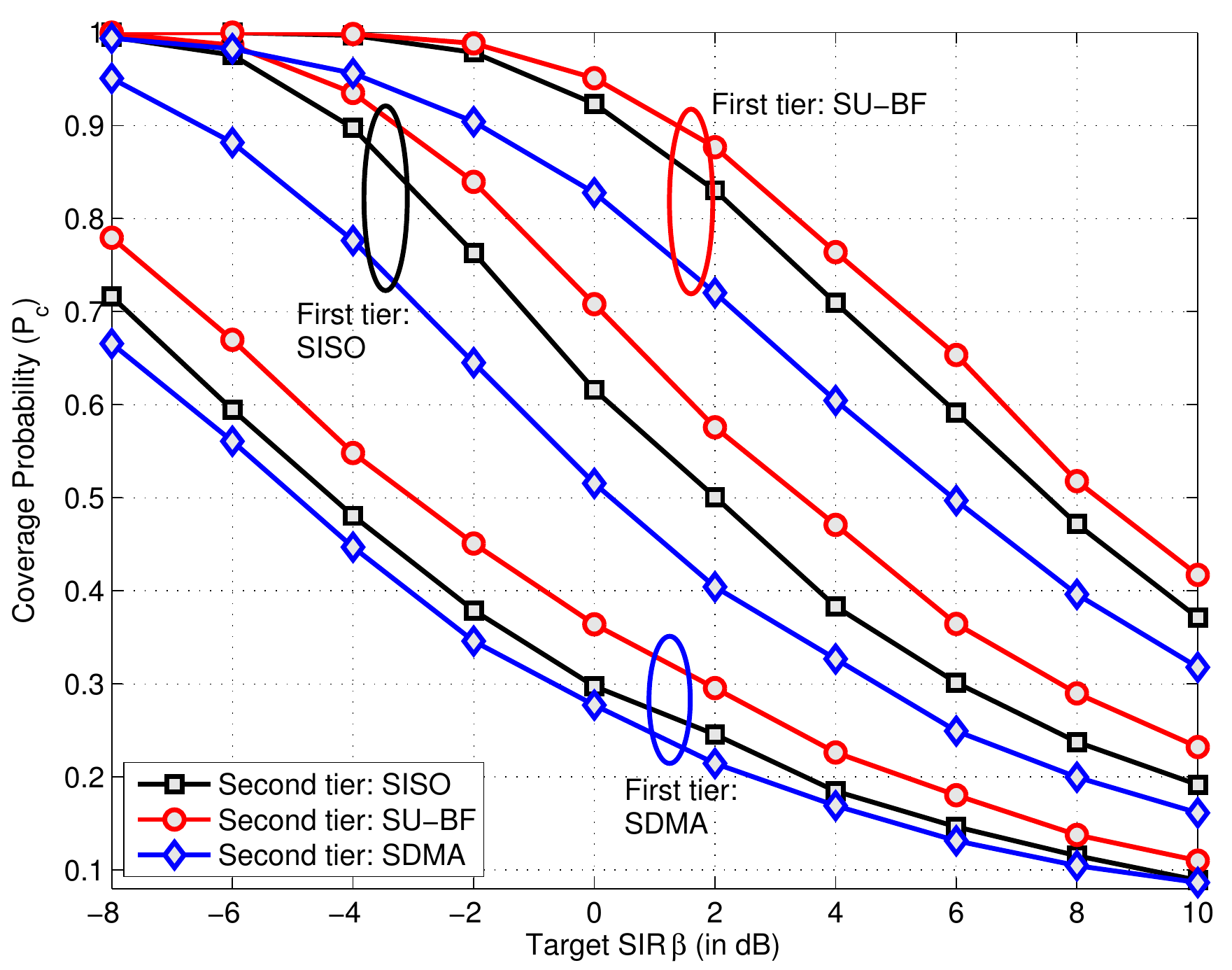}
\caption{Comparison of the coverage probability in a two-tier HetNet for various combinations of multi-antenna techniques ($K=2, P = [1, .01], \lambda_2 = 2\lambda_1, \beta_1 = \beta_2, \alpha = 3.8$). The number of antennas in case of multi-antenna tiers is $M = 4$.}
\label{fig:TierAdd_Pc}
\end{figure}

\begin{figure}[t]
\centering
\includegraphics[width=\columnwidth]{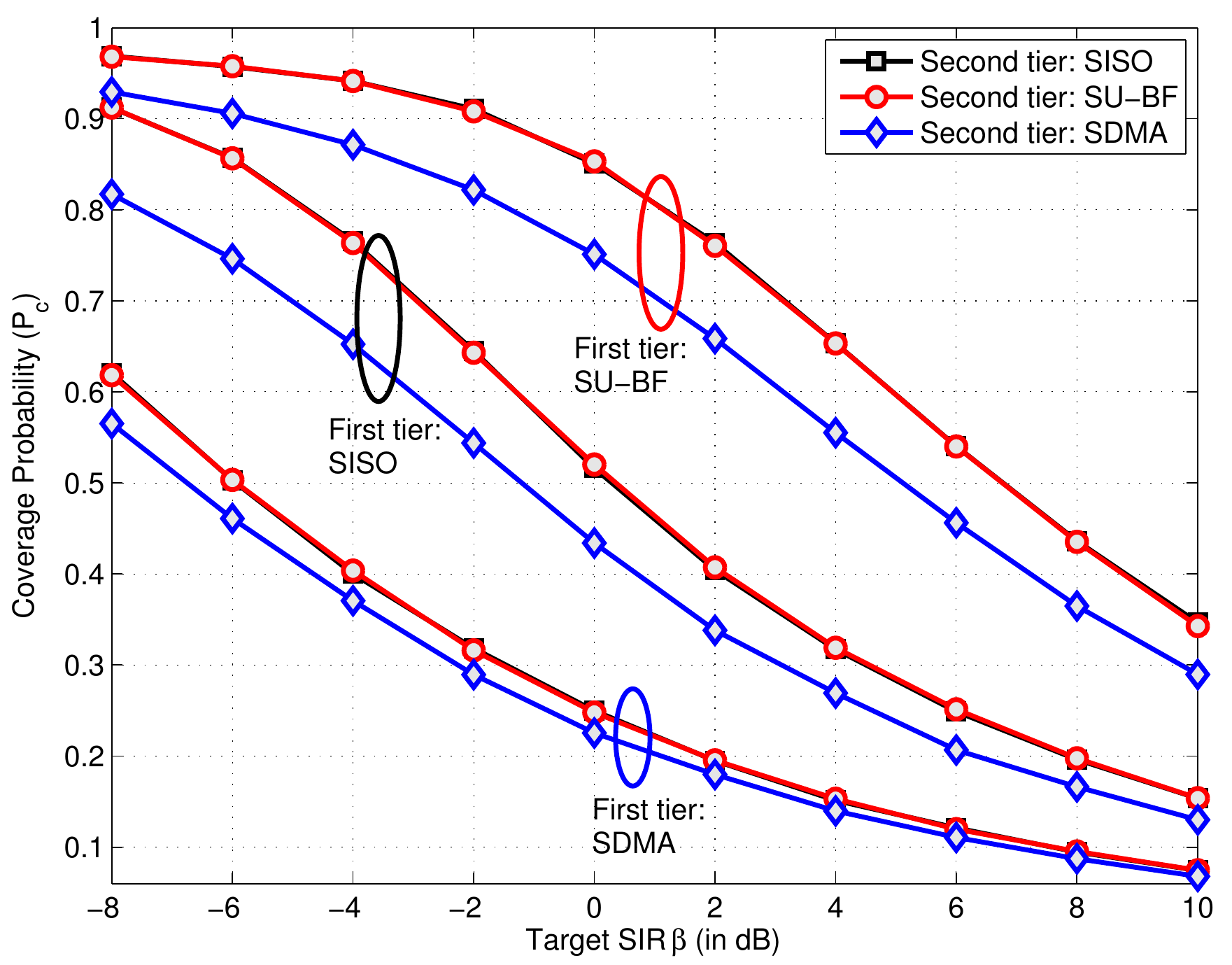}
\caption{Comparison of the coverage probability in a twotier HetNet where the second tier is in closed access ($K=2, P = [1, .01], \lambda_2 = 2\lambda_1, \beta_1 = \beta_2, \alpha = 3.8$). The number of antennas in case of multi-antenna tiers is $M = 4$.}
\label{fig:TierAdd_Pc_CA}
\end{figure}

\subsection{Effect of adding additional tier on coverage probability}
We study the effect of adding a second tier on the coverage probability of a cellular network in Figs.~\ref{fig:TierAdd_Pc} and \ref{fig:TierAdd_Pc_CA}, where both the first and the second tier can be one of the three possible types: i) SISO, ii) full SDMA, iii) SU-BF. In Fig.~\ref{fig:TierAdd_Pc}, we assume that both tiers are in open access. The result shows that the case where both tiers perform SU-BF results in the highest coverage, whereas the case where both tiers perform full SDMA leads to the lowest coverage. This is because SU-BF case has an additional beamforming gain; in addition to the proximity gain enjoyed by the SISO case. These observations are consistent with the coverage ordering results derived in Section~\ref{sec:ordering}. In Fig.~\ref{fig:TierAdd_Pc_CA}, we study the effect of adding a second tier that is in closed access, i.e., a typical user cannot connect to the second tier BSs. The performance of various transmission techniques is in the same order as for the open access case studied in Fig.~\ref{fig:TierAdd_Pc}. Interestingly, the coverage probability of a typical user is the same irrespective of whether the new closed access tier is doing SISO transmission or SU-BF. This is due to the fact that the channel power distribution of the interfering links in both the cases is $\Gamma(1,1)$, i.e., $\exp(1)$.

\subsection{Area spectral efficiency}

\begin{figure}[t]
\centering
\includegraphics[width=\columnwidth]{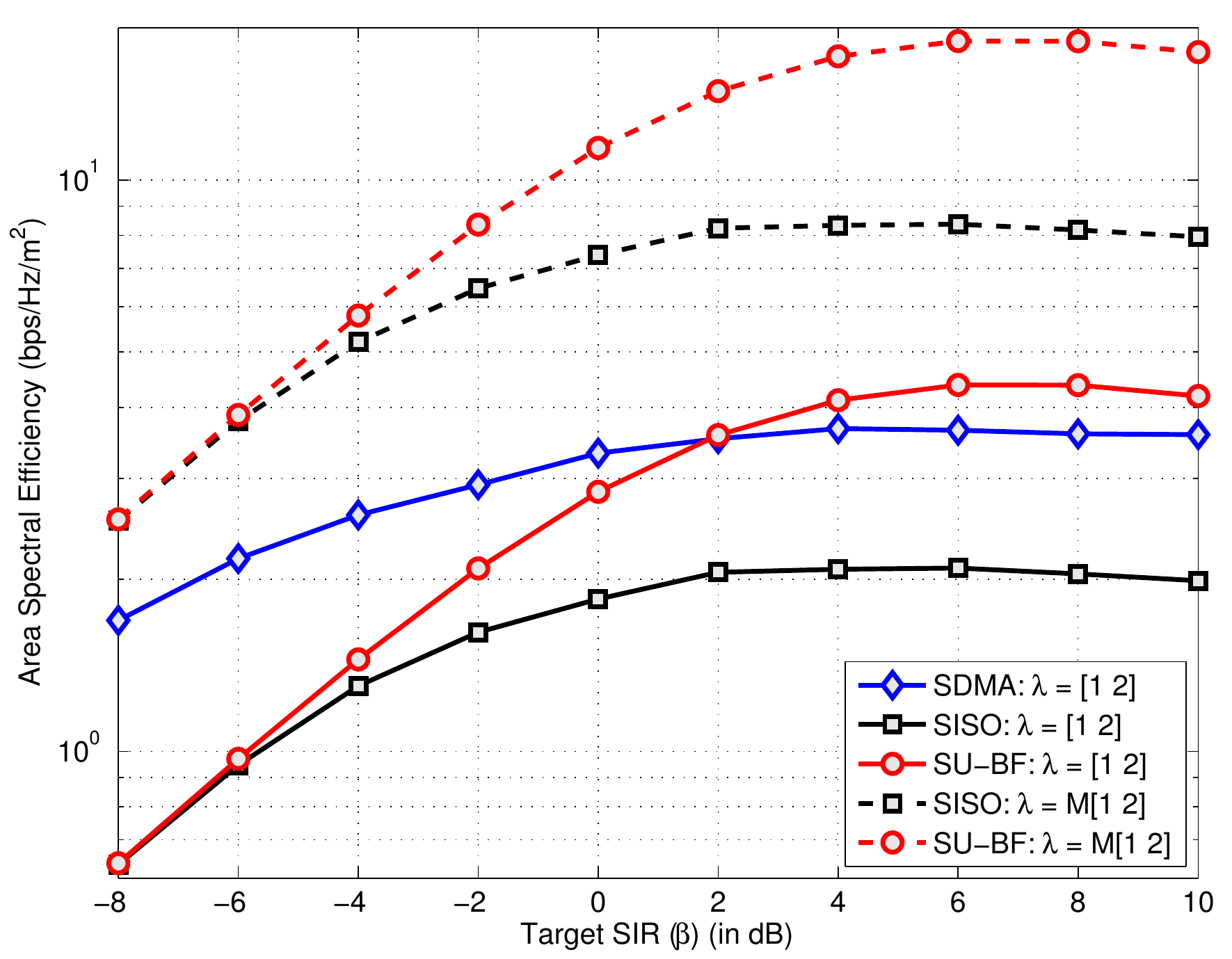}
\caption{Comparison of the ASE in a two-tier HetNet ($K=2, P = [1, .01], \lambda_2 = 2\lambda_1, \beta_1 = \beta_2, \alpha = 3.8$). The number of antennas in case of multi-antenna tiers is $M = 4$. Full SDMA corresponds to $\Psi = M$.}
\label{fig:ASE_comp}
\end{figure}

We compare the ASEs of SU-BF, SISO, and full SDMA transmission techniques in a 2-tier HetNet in Fig.~\ref{fig:ASE_comp}. Both tiers are assumed to follow the same transmission technique and the ASE result for SU-BF is computed numerically by computing the per tier coverage probability. For comparison, we consider two cases, one in which the density of the BSs in the three setups remain the same, and the other in which the densities are adjusted such that the density of users served in the three cases is the same. In the first case, SU-BF, which always outperforms SISO, even outperforms full SDMA in the high target $\sir$ regime despite serving smaller number of users. The trends in the low target $\sir$ regime are consistent with Proposition~\ref{prop:ASE}. In the second case, where the density of the users is the same in all the cases, the ordering of the three transmission techniques in terms of ASE is the same as that of coverage and rate per user.

\subsection{Effect of having a fraction of BSs in closed access}
\begin{figure}[t]
\centering
\includegraphics[width=\columnwidth]{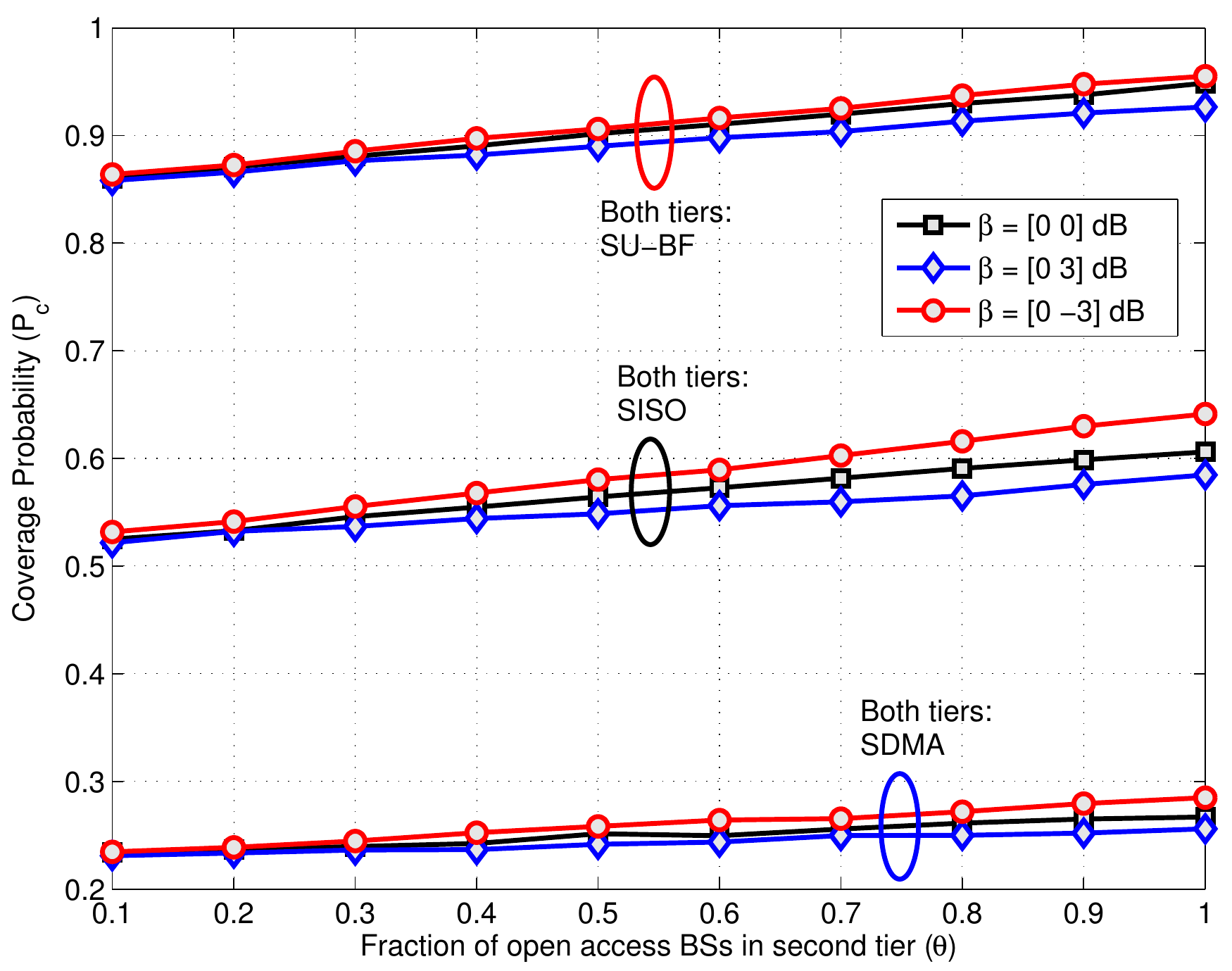}
\caption{Coverage probability when a fraction $1-\theta$ of the second tier BSs are in closed access. ($K=2, P = [1, .01], \lambda_2 = 2\lambda_1, \alpha = 3.8$).  The number of antennas in case of multi-antenna tiers is $M = 4$.}
\label{fig:PcFracCA}
\end{figure}

Before concluding this section, it is important to note that under PPP assumption, the proposed model is applicable even if a given fraction of BSs of a particular tier independently operates in open or closed access. In such a case, we can divide the original tier into two tiers with appropriate densities, which is enabled by the fact that independently thinning a PPP leads to two independent PPPs. For example, for a two-tier HetNet where all the BSs of the first tier and the fraction $0\leq \theta \leq 1$ of the second tier operate in open access while the rest in closed access, the second tier can be divided into two tiers modeled as independent PPPs $\Phi_{2^{(1)}}$ (open access tier) and $\Phi_{2^{(2)}}$ (closed access tier) with densities $\lambda_{2^{(1)}} = \theta \lambda_2$ and $\lambda_{2^{(2)}} = (1-\theta)\lambda_2$, respectively. Therefore, the original two-tier network can be reduced to an equivalent three tier network, where two tiers are in open access and one is in closed access. The numerical results for such a scenario are presented in Fig.~\ref{fig:PcFracCA} for various combinations of target-$\sir$s. Note that the coverage probability of a typical user under all considered transmission schemes increases linearly with the fraction of open access BSs $\theta$. Secondly, the loss in coverage probability with decreasing $\theta$ is higher when the target $\sir$ for the closed access BSs is lower than that of the first tier open access BSs. Similarly, the loss is lower when the target $\sir$ for the closed access BSs is higher than the first tier BSs. This is because when the target $\sir$ for the second tier is lower than the first tier, the second tier BSs would have contributed more to the coverage probability had they been in open access than the case when their target $\sir$ is higher than the first tier BSs.

\section{Conclusions}
In this paper, we have proposed a new tractable downlink model for multi-antenna HetNets. For any given BS distribution, we derived ordering results for coverage probability and per user rate to compare different transmission techniques, such as SDMA, SU-BF and baseline SISO transmission. In addition to significantly generalizing the state of the art PPP based random spatial models for cellular networks, this approach circumvents the need for deriving explicit expressions for coverage and rate, which may not reduce to simple closed forms in all the cases. Our analysis demonstrates that for a given total number of transmit antennas, it is preferable to spread them across many single-antenna BSs vs. fewer multi-antenna BSs, both in terms of coverage and rate per user. We also show that SU-BF provides higher coverage and rate per user than both SISO and SDMA due to an additional beamforming gain. To account for the fact that certain transmission techniques, such as SDMA, serve more users and may provide higher sum-rate, we derive an upper bound on the coverage probability assuming an independent PPP model for BS locations and use it to compare different transmission techniques in terms of ASE. 

Future work could consider HetNets with multi-antenna receivers and investigate potential gains by performing interference cancelation and/or receiver combining. Further extensions to this work could include the effect of opportunistic scheduling on the coverage probability and spatial reuse. Another important extension of the modeling tools developed in this paper is to the uplink of cellular networks~\cite{NovDhiJ2012}.	

\appendices

\section{Signaling Preliminaries} \label{appendix:signal}

The received signal $y_k$ from $k^{th}$ tier BS at a typical user located at the origin is given by
\begin{align}
y_{k} &= \sqrt{P_k} \|x_k\|^{-\frac{\alpha}{2}}\mathbf{v}_{kx_k}^{*}\mathbf{z}_{k} + \sum_{k\in \ncalK} \sum_{y\in \Phi_j \setminus x_k} \sqrt{P_j}\|y\|^{-\frac{\alpha}{2}}\mathbf{u}_{jy}^{*}\mathbf{z}_j,
\end{align}
where $P_k$ is the per user transmit power of $k^{th}$ tier BS, and $\mathbf{z}_k \in \mathbb{C}^{M_k \times 1}$ is the normalized transmit signal vector.
The channel vector from $k^{th}$ tier BS to a typical user located at origin is denoted by $\mathbf{v}_{kx} \in \mathbb{C}^{M_k \times 1}$ and for the interfering link from a $j^{th}$ tier BS located at $y \in \R^2$ is denoted by $\mathbf{u}_{jy} \in \mathbb{C}^{M_j \times 1}$. The vectors $\mathbf{v}, \mathbf{u}$ are assumed to have i.i.d. $\mathcal{CN}(0,1)$ entries, independent across BSs and of the user distances. 


This paper assumes linear precoding, in which the $k^{th}$ tier BS multiplies the data symbol $s_{k,i}$ destined for the $i^{th}$ user, for $1 \leq i \leq \Psi_k$, by $\mathbf{w}_{k,i}$ so that the transmitted signal is a linear function, i.e. $\mathbf{z}_k = \displaystyle \sum_{i=1}^{\Psi_k}\mathbf{w}_{k,i}s_{k,i}$.
When zero-forcing beamforming with perfect CSI is employed to serve $\Psi_k$, $\Psi_j$ users in tier $k$, $j$ respectively, the columns of the precoding matrix $\mathbf{W}_k = [\mathbf{w}_{k,i}]_{1 \leq i \leq \Psi_k} \in \mathbb{C}^{M_k \times \Psi_k}$ equal the normalized columns of $\tilde{\mathbf{V}_k}^{*}(\tilde{\mathbf{V}_k} \tilde{\mathbf{V}_k}^{*})^{-1} \in \mathbb{C}^{M_k \times \Psi_k}$, for $\tilde{\mathbf{V}} = [\tilde{\mathbf{v}}_{1}, \  \ldots, \tilde{\mathbf{v}}_{k} \ \ldots \ \tilde{\mathbf{v}}_{\Psi_k}]^{*} \in \mathbb{C}^{\Psi_k \times M_k}$ being the concatenated matrix of channel directions, where the direction of each vector channel is represented as $\tilde{\mathbf{v}}_k \triangleq \frac{\mathbf{v}_k}{\left\|\mathbf{v}_k\right\|}$.
The desired channel power is given by $h_{kx} = |\mathbf{v}_{kx}^{*} \mathbf{w}_{k,k}|^2 = |\tilde{\mathbf{v}}_{kx}^{*}\mathbf{w}_{k,k}|^2 \cdot \left\|\mathbf{v}_{kx}\right\|^2$ which equals the product of two independent rv's which are distributed as $\textrm{Beta}(M_k-\Psi_k+1,\Psi_k-1)$ and $\Gamma(M_k,1)$, respectively. 
Therefore, the channel power is $h_{kx} \sim \Gamma(\Delta_k, 1)$ with $\Delta_k = M_k - \Psi_k + 1$.
For the distribution of the interfering marks, we assume that the precoding matrices have unit-norm orthogonal columns and that $\mathbf{W}_{j}$ is calculated independently of $\mathbf{u}_{jy}$. Therefore, $\tilde{\mathbf{u}}_{jy}$ and $\mathbf{w}_{j}$ are independent isotropic unit-norm random vectors, and $\left|\tilde{\mathbf{u}}_{jy}^*\mathbf{w}_{j}\right|^2$ is a linear combination of $\Psi_j$ complex normal random variables, i.e. exponentially distributed. Neglecting the spatial correlation, we have that $g_{jy} \sim \Gamma(\Psi_j, 1)$, since it is the sum of $\Psi_j$ i.i.d. exponential random variables.

The case $\Delta_k = 1$ and $\Psi_j = M_j$ is referred to as full SDMA. The case that each BS only serves one user, i.e. $\Psi_k = 1$, using the beamforming vector $\mathbf{w}_{kx} = \tilde{\mathbf{v}}_{kx}$ corresponds to SU-BF or MISO eigen-beamforming. In that case, the channel power is given by $h_{kx} \sim \Gamma(\Delta_k, 1)$ with $\Delta_k = M_k$ and the interference marks as $g_{jy} \sim \Gamma(\Psi_j, 1)$ with $\Psi_j = 1$, $\forall j \in \ncalK$, since the beamforming vectors $\mathbf{w}_{jy}$ used by the $j^{th}$ tier interfering BS are calculated based on $\mathbf{v}_{j}$, i.e. independently of $\mathbf{u}_{jy}$.

\section{Proof of Lemma~\ref{lem:laplace}}  \label{appendix:Laplace}
The Laplace transform of interference $\mathcal{L}_{I}(s) = \E_{I_{x_k}} \left[ e^{-s I} \right]$ can be derived as
\begin{align}
&\E_{I} \left[ e^{-s I} \right] = \E_{I} \left[ e^{-s \sum_{j \in \ncalK} \sum_{y \in \Phi_j} P_j g_{jy} \|y\|^{-\alpha}
} \right]\\ 
&\stackrel{(a)}{=} \prod_{j\in \ncalK} \E \left[ \prod_{y \in \Phi_j} e^{-s P_j g_{jy} \|y\|^{-\alpha} } \right]\\
&\stackrel{(b)}{=} \prod_{j\in \ncalK} \E_{\Phi_j} \left[ \prod_{y \in \Phi_j } \mathcal{L}_{g_{jy}}\left(s P_j \|y\|^{-\alpha} \right) \right]\\
&\stackrel{(c)}{=} \prod_{j\in \ncalK} \exp \left( -\lambda_j \int_{\R^2} \left( 1 -  \mathcal{L}_{g_{jy}}\left(s P_j \|y\|^{-\alpha} \right) \right) \nrmd  y\right)\\
&\stackrel{(d)}{=} \prod_{j\in \ncalK} \exp \left( -\lambda_j \int_{\R^2} \left( 1 -  \frac{1}{(1+ s P_j \|y\|^{-\alpha})^{\Psi_j}} \right) \nrmd  y\right)\\
&= \prod_{j\in \ncalK} \exp \left( -\lambda_j \int_{\R^2} \frac{(1+ s P_j \|y\|^{-\alpha})^{\Psi_j}-1}{(1+ s P_j \|y\|^{-\alpha})^{\Psi_j}} \nrmd  y\right)\\
&\stackrel{(e)}{=} \prod_{j\in \ncalK} \exp \left( -\lambda_j \int_{\R^2} \frac{\sum_{m=1}^{\Psi_j} {\Psi_j \choose m} (s P_j \|y\|^{-\alpha})^m }{(1+ s P_j \|y\|^{-\alpha})^{\Psi_j}} \nrmd  y\right)\\
&= \prod_{j\in \ncalK} \exp \left( -\lambda_j \sum_{m=1}^{\Psi_j} {\Psi_j \choose m}  \int_{\R^2} \frac{(s_{x_k} P_j \|y\|^{-\alpha})^m }{(1+ s_{x_k} P_j \|y\|^{-\alpha})^{\Psi_j}} \nrmd  y\right)\\
&\stackrel{(f)}{=} \prod_{j\in \ncalK} e^{ -2 \pi \lambda_j (s_{x_k} P_j)^\frac{2}{\alpha} \sum_{m=1}^{\Psi_j} {\Psi_j \choose m}  \int_{0}^{\infty} \frac{r^{-\alpha m} }{(1+ r^{-\alpha})^{\Psi_j}}  r\nrmd  r}\\
&\stackrel{(g)}{=} \exp\left( - s_{x_k}^{\frac{2}{\alpha}} \sum_{j \in \ncalK} \lambda_j P_j^{\frac{2}{\alpha}}   C (\alpha, \Psi_j) \right),
\end{align}
where $(a)$ follows from the independence of the tiers, $(b)$ follows from the fact that channel powers are independent of the BS locations, $(c)$ follows from PGFL of PPP~\cite{StoKenB1995}, $(d)$ follows from the Laplace transform of the $g_{jy} \sim \Gamma(\Psi_j, 1)$, $(e)$ follows from Binomial theorem, and $(f)$ follows from converting to Cartesian to polar coordinates, and $(g)$ follows by substituting $(1+r^{-\alpha})^{-1} \rightarrow t$ to convert the integral into Euler's Beta function $B(x,y) = \int_{0}^1 t^{x-1} (1-t)^{y-1} \nrmd  t$.

\bibliographystyle{IEEEtran}
\bibliography{arXiv130805}
\end{document}